\documentclass[12pt]{article}
\usepackage[latin9]{inputenc}
\usepackage[a4paper]{geometry}
\usepackage{float}
\usepackage{mathtools}
\usepackage{amsmath}
\usepackage{amsthm}
\usepackage{amssymb}
\usepackage{esint}
\usepackage[numbers]{natbib}
\usepackage[unicode=true,
 bookmarks=true,bookmarksnumbered=false,bookmarksopen=false,
 breaklinks=false,pdfborder={0 0 0},pdfborderstyle={},backref=false,colorlinks=false]
 {hyperref}

\makeatletter
\numberwithin{equation}{section}
\numberwithin{figure}{section}
\theoremstyle{plain}
\newtheorem{thm}{\protect\theoremname}
\theoremstyle{definition}
\newtheorem{defn}[thm]{\protect\definitionname}
\theoremstyle{plain}
\newtheorem{lem}[thm]{\protect\lemmaname}
\theoremstyle{plain}
\newtheorem{prop}[thm]{\protect\propositionname}
\theoremstyle{plain}
\newtheorem{cor}[thm]{\protect\corollaryname}

\usepackage{braket}
\usepackage{feynmf}

\makeatother

\providecommand{\corollaryname}{Corollary}
\providecommand{\definitionname}{Definition}
\providecommand{\lemmaname}{Lemma}
\providecommand{\propositionname}{Proposition}
\providecommand{\theoremname}{Theorem}

\begin{document}


\begin{titlepage}

\thispagestyle{empty}

\begin{flushright}
YITP-18-41\\
\end{flushright}

\vspace{.4cm}
\begin{center}
\noindent{\large \bf Entanglement of Purification for Multipartite States\\ and its Holographic Dual}\\
\vspace{2cm}

Koji Umemoto$^{a}$ and Yang Zhou$^{b}$
\vspace{1cm}

{\it
$^{a}$Center for Gravitational Physics, \\
Yukawa Institute for Theoretical Physics (YITP), Kyoto University, \\
Kitashirakawa Oiwakecho, Sakyo-ku, Kyoto 606-8502, Japan\\ \vspace{3mm}
$^{b}$Department of Physics and Center for Field Theory and Particle Physics, \\
Fudan University, Shanghai 200433, China\\
}

\end{center}

\vspace{.5cm}
\begin{abstract}
We introduce a new information-theoretic measure of multipartite correlations $\Delta_P$, by generalizing the entanglement of purification to multipartite states. We provide proofs of its various properties, focusing on several entropic inequalities, in generic quantum systems. In particular, it turns out that the multipartite entanglement of purification gives an upper bound on multipartite mutual information, which is a generalization of quantum mutual information in the spirit of relative entropy. After that, motivated by a tensor network description of the AdS/CFT correspondence, we conjecture a holographic dual of multipartite entanglement of purification $\Delta_W$, as a sum of minimal areas of codimension-2 surfaces which divide the entanglement wedge into multi-pieces. We prove that this geometrical quantity satisfies all properties we proved for the multipartite entanglement of purification. These agreements strongly support the $\Delta_{P}=\Delta_{W}$ conjecture. We also show that the multipartite entanglement of purification is larger than multipartite squashed entanglement, which is a promising measure of multipartite quantum entanglement. We discuss potential saturation of multipartite squashed entanglement onto multipartite mutual information in holographic CFTs and its applications.

\end{abstract}

\end{titlepage}

\newpage
\setcounter{tocdepth}{2}
\tableofcontents
\section{Introduction}

Quantum entanglement has recently played significant roles in condensed
matter physics~\citep{Vidal:2002rm,Kitaev:2005dm,Levin:2006}, particle
physics~\citep{Casini:2004bw,Calabrese:2004eu,Casini:2012ei,Casini:2008cr,Hung:2018rhg}
and string theory~\citep{RT:RT-formula,HRT:HRT-formula}. To study
the quantum entanglement in states of a quantum system, we often
divide the total system into a subsystem $A$ and its complement $A^{c}$,
and then compute the entanglement entropy $S_{A}:=-{\rm Tr}\rho_{A}\log\rho_{A}$,
where $\rho_{A}$ is the reduced density matrix of a given total state
$\rho_{A}:={\rm Tr}_{A^{c}}\rho_{tot}$. In the gauge/gravity correspondence
\citep{Maldacena:AdS/CFT}, the Ryu-Takayanagi formula \citep{RT:RT-formula,HRT:HRT-formula}
allows one to compute the entanglement entropy in CFTs by a minimal
area of codimension-2 surface in AdS. This discovery opens a new era
of studying precise relations between spacetime geometry and quantum
entanglement \citep{Swingle:TensorNetwork,VanRaamsdonk:Buildingup1,MaldacenaSusskind:ERequalsEPR,FGHMR:LinearEinsteinequationfrom1stLaw,MiyajiTakayanagi:SurfaceStateCorrespondence,MTW:OptimizationofPathIntegral1,CKMTW:OptimizationofPathIntegral2,PYHP:HaPPYcode,FreedmanHeadrick:BitThreads}.

Recently the holographic counterpart of a new quantity independent
of entanglement entropy, called the entanglement of purification $E_{P}$
\citep{THLD:EoP}, has been conjectured \citep{TU:HEoP,NDHZS:HEoP}.
The entanglement of purification quantifies an amount of total correlation,
including quantum entanglement, for bipartite mixed states $\rho_{AB}$
acting on $\mathcal{H}_{AB}\equiv\mathcal{H}_{A}\otimes\mathcal{H}_{B}$.
On the other hand, the entanglement entropy truly measures quantum entanglement
only for pure states $\ket{\psi}_{AB}$. Indeed, it is not even a
correlation measure if a given total state is mixed. Thus, so
called $E_{P}=E_{W}$ conjecture fills a gap between correlations in mixed states and spacetime geometry.
This has been further studied in the literature \citep{BaoHalpern:HEoPgeneralization,ATU:EoPfreeQFTs,HiraiTamaokaYokoya:TowardsEoPinCFTs,EGP:WEReconstructionEoP,NRS:PullingBoundaryintoBulk}.

Another important direction is to explore multipartite correlation measure and its geometric dual.
It is well known that there are richer correlation structures in quantum
systems consisting of three or more subsystems, especially about quantum
entanglement (see e.g. \citep{HHHH:QuantumEntanglement}). For example,
there appear different separability criteria for multipartite states.
To understand operational aspects of an amount of entanglement
by means of (S)LOCC, one needs (infinitely) many kinds of standard
states, such as W and GHZ states. Therefore multipartite entanglement
is much more sophisticated than bipartite entanglement. 

On the other hand, the holographic interpretation of multipartite
correlations is less known, though it can play a crucial role to understand the emergence
of bulk geometry out of boundary renormalization group flows as well
as the idea of ER=EPR \citep{MaldacenaSusskind:ERequalsEPR,HHM:MonogamyofHMI,BHMMR:MultiboundaryWormholes,NezamiWalter:MultiEntinStabilizerTN}. In the literature, one quantity $\tilde{I}_{3}$
called tripartite information, which is one particular generalization of mutual information $I(A:B)=S_{A}+S_{B}-S_{AB}$, attracts
a lot of attention in holography \citep{HHM:MonogamyofHMI,HQRY:ChaosinQCs,Rota:TripartiteMI,MTV:NpartiteHMI,AMT:NpartiteHMI}.
However, it can be either positive or negative, which makes it a bit
hard to use it to diagnose physical properties of the system at times. 

In this paper, we first introduce a new measure of multipartite quantum
and classical correlations $\Delta_{P}$ in generic quantum systems.
It is given by generalizing the entanglement of purification to multipartite
states. We study its information-theoretic properties, especially
focusing on various entropic inequalities. It turns out that $\Delta_{P}$
gives an upper bound on another generalization of quantum mutual information introduced in \citep{YHHHOS:mEsq,AHS:mEsq}, the so called multipartite mutual information.

We then propose a holographic dual of multipartite entanglement of
purification $\Delta_{W}$, as a sum of minimal areas of codimension-2
surfaces in the entanglement wedge \citep{CKNR:EntanglementWedge,Wall:Maximinsurfaces,HHLR:EntanglementWedge},
motivated by the tensor network description of the AdS/CFT correspondence
\citep{Swingle:TensorNetwork,MiyajiTakayanagi:SurfaceStateCorrespondence,MTW:OptimizationofPathIntegral1,CKMTW:OptimizationofPathIntegral2}.
One typical example of $\Delta_{W}$ was drawn in Fig.\ref{fig:MEoP0}. We demonstrate that it satisfies
all the properties of the multipartite entanglement of purification
by geometrical proofs. These agreements tempt us to provide a new
conjecture $\Delta_{P}=\Delta_{W}$.

\begin{figure}[t]
\centering{}\includegraphics[scale=0.2]{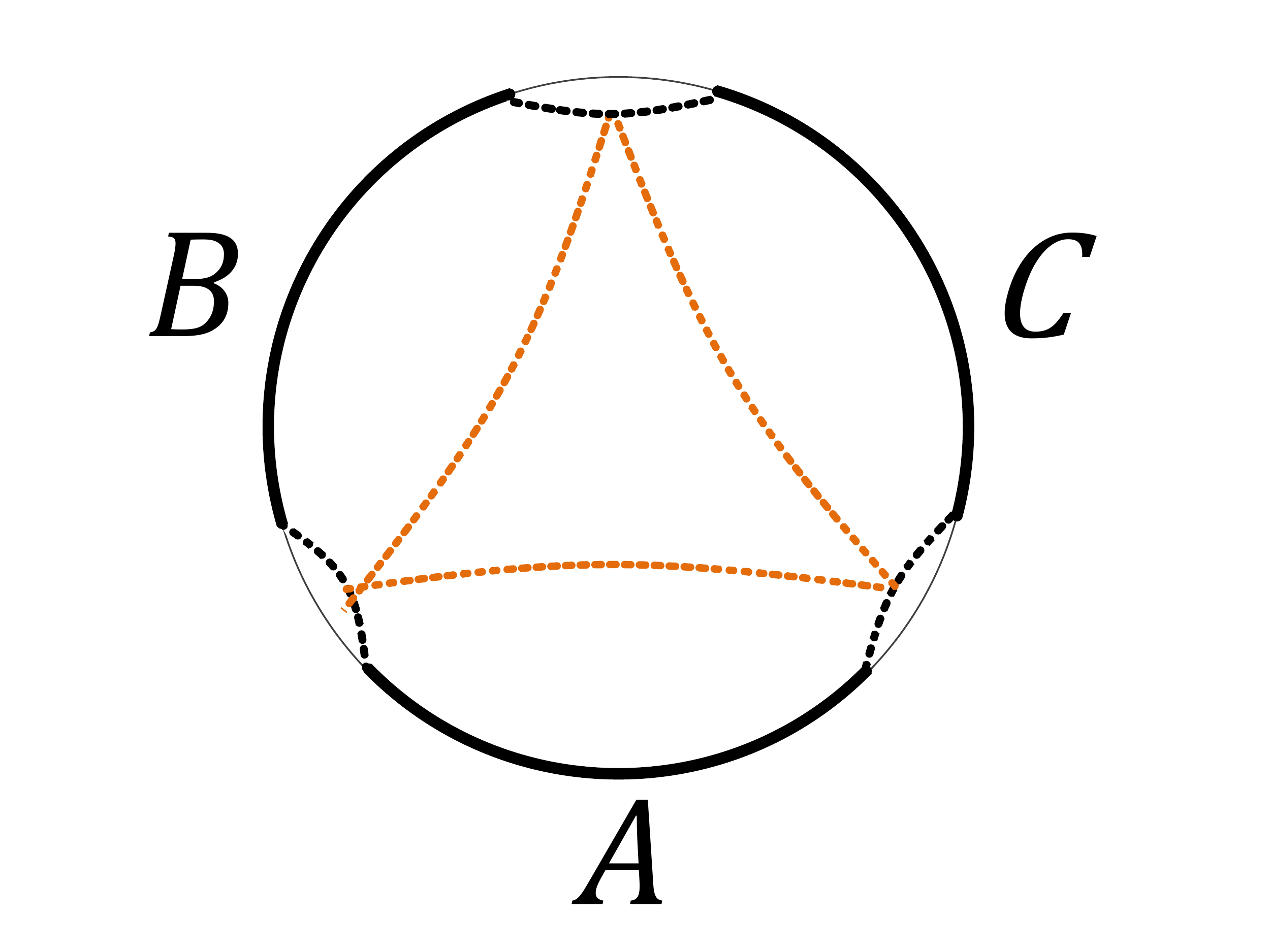}\caption{\label{fig:MEoP0} An example of minimal surfaces which gives $\Delta_{W}$
for a tripartite setup.}
\end{figure}

We clarify in the appendix that $\Delta_P$ is larger than the multipartite generalization \citep{YHHHOS:mEsq,AHS:mEsq} of squashed entanglement, which is a promising measure of genuine quantum entanglement \citep{Tucci:Esq,ChristandlWinter:SquashedEntanglement,KoashiWinter:monogamy}. We also discuss the holographic counterpart of squashed entanglement.


This paper is organized as follows: In Section \ref{sec2}, we give
the definition of multipartite entanglement of purification $\Delta_{P}$
and prove its information-theoretic properties. In Section \ref{sec3},
we introduce a multipartite generalization of the entanglement wedge
cross-section $\Delta_{W}$ in holography and study its properties geometrically, and find full agreement with those
we proved in Section \ref{sec2}. Based on these facts, we propose the conjecture
$\Delta_{P}=\Delta_{W}$. In Section \ref{sec5}, we conclude and discuss future questions. In appendix \ref{sec4}, the universal
relation between $\Delta_{P}$ and multipartite squashed entanglement
was clarified and the holographic counterpart of squashed entanglement was also discussed. \\

\textit{Note added}: After all the results in this paper were obtained, \citep{BaoHalpern:multipartiteEoP} appeared in which the authors defined the same multipartite
generalization of $E_{P}$ (up to a normalization factor) and tested a holographic dual which is different from $\Delta_{W}$.

\section{Multipartite entanglement of purification\label{sec2}}

In this section, we will define a generalization of entanglement of
purification for multipartite correlations and prove its various information-theoretic
properties.

Let us start from recalling the definition and basic properties of
the entanglement of purification \citep{THLD:EoP,BagchiPati:EoPproperties}. We consider a quantum state
$\rho_{AB}$ on a bipartite quantum system $\mathcal{H}_{AB}=\mathcal{H}_{A}\otimes\mathcal{H}_{B}$.
The entanglement of purification $E_{P}$ for a bipartite state $\rho_{AB}$
is defined by 
\begin{equation}
E_{P}(\rho_{AB}):=\min_{\ket{\psi}_{AA'BB'}}S_{AA'},\label{eq:Def EoP}
\end{equation}
where the minimization is taken over purifications $\rho_{AB}={\rm Tr}_{A'B'}[\ket{\psi}\bra{\psi}_{AA'BB'}]$.
This is an information-theoretic measure of total correlations, namely,
it captures both quantum and classical correlations between $A$
and $B$. In this point of view, it is similar as the quantum mutual information
$I(A:B)=S_{A}+S_{B}-S_{AB}$. Nevertheless, it is known that the regularized
version of entanglement of purification $E_{P}^{\infty}(\rho_{AB}):=\lim_{n\to\infty}E_{P}(\rho_{AB}^{\otimes n})/n$
has an operational interpretation in terms of EPR pairs and local
operation and asymptotically vanishing communication \citep{THLD:EoP}.
From the definition, we can see that $E_{P}$ measures the value of
quantum entanglement between $AA'$ and $BB'$ in the optimally purified
system.

We simply write $E_{P}(\rho_{AB})$ as $E_{P}(A:B)$ unless otherwise a
given state specified. We summarize some known properties of
$E_{P}$ for the reader's convenience:

(I) It reduces to the entanglement entropy for pure state $\rho_{AB}=\ket{\psi}\bra{\psi}_{AB}$,
\begin{equation}
E_{P}(A:B)=S_{A}=S_{B},\ {\rm {\rm for\ pure\ states}}.
\end{equation}

(II) It vanishes if and only if a given state $\rho_{AB}$ is a product, 
\begin{equation}
E_{P}(A:B)=0\ {\rm \Leftrightarrow}\ \rho_{AB}=\rho_{A}\otimes\rho_{B}.
\end{equation}

(III) It monotonically decreases upon discarding ancilla, 
\begin{equation}
E_{P}(A:BC)\geq E_{P}(A:B).
\end{equation}

(IV) It is bounded from above by the entanglement entropy,
\begin{equation}
E_{P}(A:B)\leq\min\{S_{A},S_{B}\}.
\end{equation}

(Va) It is bounded from below by a half of mutual information,
\begin{equation}
E_{P}(A:B)\geq\frac{I(A:B)}{2}.
\end{equation}

(Vb) For tripartite state $\rho_{ABC}$, it has a lower bound,
\begin{equation}
E_{P}(A:BC)\geq\frac{I(A:B)+I(A:C)}{2}.
\end{equation}

(VI) For tripartite pure state $\ket{\psi}_{ABC}$, it is polygamous,
\begin{equation}
E_{P}(A:BC)\leq E_{P}(A:B)+E_{P}(A:C).
\end{equation}

(VIIa) For a class of states that saturate the subadditivity, i.e.
$S_{AB}=S_{B}-S_{A}$, it reduces to the entanglement entropy, 
\begin{equation}
E_{P}(A:B)=S_{A}\ {\rm when}\ S_{AB}=S_{B}-S_{A}.
\end{equation}

(VIIb) For a class of states that saturate the strong subadditivity,
i.e. $S_{AB}+S_{AC}=S_{B}+S_{C}$, it reduces to the entanglement
entropy, 
\begin{equation}
E_{P}(A:B)=S_{A}\ {\rm when}\ S_{AB}+S_{AC}=S_{B}+S_{C}.
\end{equation}
These properties are not independent from each other and one can prove
(VI) from (I) and (Va), and also prove (VIIa) (or (VIIb)) from (III)
and (Va) (or (Vb)). All these properties are proven
in generic quantum systems. (I) allows us to regard $E_{P}$
as a generalization of entanglement entropy for quantifying an amount
of correlations for mixed states. We refer to \citep{THLD:EoP,BagchiPati:EoPproperties}
for detailed proofs and discussion.

\subsection{Definition}

We define a new quantity that captures total multipartite correlations
by generalizing the entanglement of purification. We first note that
we can rewrite the definition of $E_{P}$ \eqref{eq:Def EoP} as 
\begin{equation}
E_{P}(\rho_{AB})=\frac{1}{2}\min_{\ket{\psi}_{AA'BB'}}[S_{AA'}+S_{BB'}],
\end{equation}
because $S_{AA'}=S_{BB'}$ holds for any purifications $\ket{\psi}_{AA'BB'}$.
This form of $E_{P}$ motivate us to define a generalization of entanglement
of purification for a $n$-partite state $\rho_{A_{1}\cdots A_{n}}$
as follows. 
\begin{defn}
For $n$-partite quantum states $\rho_{A_{1}\cdots A_{n}}$, we define
the multipartite entanglement of purification $\Delta_{P}$ as \footnote{The normalization factor is not essential in our discussion, so we take it so that the entropic inequalities below become simple. We remark that it is common and even operationally meaningful at times in quantum information theory to define a multipartite generalization of some measure with $\frac{1}{2}$ prefactor regardless to $n$. If we follow this convention for $\Delta_P$, some results in this paper will seemingly change e.g. from $\Delta_P \geq I$ to $\Delta_P \geq I/2$.} 
\begin{equation}
\Delta_{P}(\rho_{A_{1}:\cdots:A_{n}})\coloneqq\min_{\ket{\psi}_{A_{1}A_{1}'\cdots A_{n}A_{n}'}}\sum_{i=1}^{n}S_{A_{i}A_{i}'},
\end{equation}
where the minimization is taken over all possible purifications of
$\rho_{A_{1}\cdots A_{n}}$. 
\end{defn}

We call it multipartite entanglement of purification and write
$\Delta_{P}(\rho_{A_{1}:\cdots:A_{n}})=\Delta_{P}(A_{1}:\cdots:A_{n})$
unless we need to specify a given state. $\Delta_{P}$ can be
regarded as the value of sum of quantum entanglement
in an optimal purification $\ket{\psi}_{A_{1}A_{1}'\cdots A_{n}A_{n}'}$,
between one of $n$-parties and the other $n-1$ parts. For example,
for tripartite state $\rho_{ABC}$, it can be represented as 
\begin{equation}\label{def3}
\Delta_{P}(A:B:C)\coloneqq\min_{\ket{\psi}_{AA'BB'CC'}}[S_{AA'}+S_{BB'}+S_{CC'}],
\end{equation}
and the entanglement entropies $S_{AA'},\ S_{BB'}$ and $S_{CC'}$
in the brackets characterize quantum entanglement between $AA':BB'CC'$,
$BB':AA'CC'$ and $CC':AA'BB'$, respectively. We also note that a
purification that gives the optimal value in (\ref{def3}) may not be unique in general,
as is so for the entanglement of purification $E_{P}$.

\subsection{Other measures}
Besides the entanglement of purification, there have been a lot of
measures of quantum and/or classical correlations which quantify
bipartite correlations for mixed states - including mutual information
and squashed entanglement \citep{Tucci:Esq,ChristandlWinter:SquashedEntanglement}.
Their generalization for multipartite cases have also been
proposed in the literature \citep{YHHHOS:mEsq,AHS:mEsq}. To see that multipartite
entanglement is not just a sum of bipartite one, let us consider
the famous GHZ state in a 3-qubit system: 
\begin{equation}
\ket{\rm GHZ}=\frac{1}{\sqrt{2}}(\ket{000}+\ket{111}).
\end{equation}
For this state, the three qubits are strongly correlated by genuine tripartite entanglement. The point is that, after one of the subsystems is traced
out, the remaining bipartite state is just a separable state and there
is no quantum entanglement. This is an example which shows that the structure
of multipartite quantum entanglement is much richer than bipartite ones. 

To quantify an amount of correlations for multipartite
states, there have been several proposals by generalizing the
mutual information. One well known quantity, the so-called tripartite information is defined as (based on Ven diagram) 
\begin{align}
\tilde{I}_{3}(A:B:C) & :=S_{A}+S_{B}+S_{C}-S_{AB}-S_{BC}-S_{CA}+S_{ABC}\nonumber \\
 & =I(A:C)+I(A:B)-I(A:BC).
\end{align}
This quantity plays an important role in holography
and has been studied e.g. in \citep{HHM:MonogamyofHMI,HQRY:ChaosinQCs,Rota:TripartiteMI,MTV:NpartiteHMI}. The tripartite information $\tilde{I}_{3}$, however, can take all of negative, zero, or positive values in generic quantum systems, and especially it always vanishes for all pure states.
These facts make it hard to be regarded as a standard measure
of correlations.

An alternative approach is introduced in \citep{YHHHOS:mEsq,AHS:mEsq}
and the authors called the generalization as multipartite mutual information,
defined by the relative entropy between a given original state and its local product state: 
\begin{equation}
I(A:B:C):=S(\rho_{ABC}||\rho_{A}\otimes\rho_{B}\otimes\rho_{C})=S_{A}+S_{B}+S_{C}-S_{ABC},
\end{equation}
where $S(\rho||\sigma)={\rm Tr}\rho\log\rho-{\rm Tr}\rho\log\sigma$. This definition is motivated by an expression of the bipartite mutual
information, 
\begin{equation}
I(A:B)=S(\rho_{AB}||\rho_{A}\otimes\rho_{B}).
\end{equation}
In general,
the multipartite mutual information for $n$-partite state $\rho_{A_{1}\cdots A_{n}}$is
defined in the same manner: 
\begin{equation}
I(A_{1}:\cdots:A_{n}):=S(\rho_{A_{1}\cdots A_{n}}||\rho_{A_{1}}\otimes\cdots\otimes\rho_{A_{n}})=\sum_{i=1}^{n}S_{A_{i}}-S_{A_{1}\cdots A_{n}}.
\end{equation}
This is clearly positive semi-definite and is monotonic under local
operations \citep{YHHHOS:mEsq}. Thus one may consider it as a more promising measure than $\tilde{I}$ for the purpose of total multipartite correlations. Note that we can
rewrite $I(A_{1}:\cdots:A_{n})$ as a suggestive form, which is a summation
of bipartite mutual information: 
\begin{equation}
I(A_{1}:\cdots:A_{n})=I(A_{1}:A_{2})+I(A_{1}A_{2}:A_{3})+\cdots+I(A_{1}\cdots A_{n-1}:A_{n}).
\end{equation}
We will carefully distinguish these two types of generalizations of
mutual information, $\tilde{I}$ and $I$, and we mainly discuss the latter one in the following context.

\subsection{Properties of multipartite entanglement of purification}

Here we explore properties of $\Delta_P$ in generic quantum systems. Firstly, we can immediately see the following property from the definition. 
\begin{lem}
\label{lem:Delta decoupling relations}If one of the subsystems is
decoupled $\rho_{A_{1}\cdots A_{n}}=\rho_{A_{1}\cdots A_{n-1}}\otimes\rho_{A_{n}}$,
then $\Delta_{P}(A_{1}:\cdots:A_{n-1}:A_{n})=\Delta_{P}(A_{1}:\cdots:A_{n-1})$. 
\end{lem}

\begin{proof}
One can separately purify $\rho_{A_{n}}$ and the remaining parts,
and then it directly follows by definition. 
\end{proof}
Especially, for bipartite state $\rho_{AB}$, it reduces to the twice
of entanglement of purification: $\Delta_{P}(A:B)=2E_{P}(A:B)$. This guarantees that $\Delta_{P}$ is a generalization of $E_{P}$
to multipartite states.

We expect that $\Delta_{P}$ is a natural generalization of $E_{P}$
for multipartite correlations and has similar properties.
Indeed, one can prove the following properties which are the counterparts
of those of $E_{P}$ mentioned above. 
\begin{prop}
\label{prop:Delta for pure}If a given $n$-partite state
$\ket{\phi}_{A_{1}\cdots A_{n}}$ is pure, then the multipartite entanglement
of purification is given by the summation of entanglement entropy
of each single subsystem, 
\begin{equation}
\Delta_{P}(A_{1}:\cdots:A_{n})=\sum_{i=1}^{n}S_{A_{i}}\ {\rm for\ pure\ states.}
\end{equation}
\end{prop}

\begin{proof}
Notice that $\ket{\phi}_{A_{1}\cdots A_{n}}$ itself is a purification,
and that all the other purifications should have a form $\ket{\psi}_{A_{1}A_{1}'\cdots A_{n}A_{n}'}=\ket{\phi}_{A_{1}\cdots A_{n}}\otimes\ket{\phi'}_{A_{1}'\cdots A_{n'}'}$.
Thus adding ancillary systems always increases the sum of entanglement
entropy of purified systems, and the minimum is achieved by
the original state $\ket{\psi}_{A_{1}A_{1}'\cdots A_{n}A_{n}'}=\ket{\phi}_{A_{1}\cdots A_{n}}$. 
\end{proof}
This is a generalization of property (I) and makes it easy to calculate
$\Delta_{P}$ for pure states. Note that in the case of pure state,
the multipartite mutual information also reduces to the sum of entanglement
entropy: $I(A_{1}:\cdots:A_{n})=\sum_{i=1}^{n}S_{A_{i}}-S_{A_{1}\cdots A_{n}}=\sum_{i=1}^{n}S_{A_{i}}$.
Thus we have $\Delta_{P}=I$ for any pure multipartite states. 
\begin{prop}
\label{prop:Fauthfullness of Delta }$\Delta_{P}$ vanishes if and
only if a given $n$-partite state is fully product, 
\begin{equation}
\Delta_{p}(A_{1}:\cdots:A_{n})=0\Leftrightarrow\rho_{A_{1}\cdots A_{n}}=\rho_{A_{1}}\otimes\cdots\otimes\rho_{A_{n}}.
\end{equation}
\end{prop}

Even though it can be directly proven, we postpone its proof after
the proposition \ref{prop:n-partite Delta > MI}. This property indicates
that $\Delta_{P}$ is not a measure of merely quantum
entanglement but of both quantum and classical correlations.
This is expected since it is a generalization of $E_{P}$.
Note that $\Delta_{P}$ is not a measure of genuine $n$-partite
correlations, but also includes $2,\cdots,n-1$-partite correlations in $\rho_{A_{1}\cdots A_{n}}$.\\

As a measure of correlations, it is natural to expect that $\Delta_{P}$ decreases when we trace
out a part of one subsystem. Actually, this is true as we can see in the following. 
\begin{prop}
\label{prop:Delta monotonicity}$\Delta_{P}$ monotonically decrease
upon discarding ancilla, 
\begin{equation}
\Delta_{P}(XA_{1}:\cdots:A_{n})\geq\Delta_{P}(A_{1}:\cdots:A_{n}).
\end{equation}
\end{prop}

\begin{proof}
It follows from the fact that all purifications of $\rho_{XA_{1}\cdots A_{n}}$
are included in these of $\rho_{A_{1}\cdots A_{n}}$. Namely, if $\ket{\psi}_{XX'A_{1}A_{1}'\cdots A_{n}A_{n}'}$
is a optimal purification for $\rho_{XA_{1}\cdots A_{n}},$ then it
is also one of the (not optimal in general) purification of
$\rho_{A_{1}\cdots A_{n}}$, thus 
\begin{equation}
\Delta_{P}(XA_{1}:\cdots:A_{n})=S_{A_{1}(A_{1}'XX')}+S_{A_{2}A_{2}'}+\cdots+S_{A_{n}A_{n}'}\geq\Delta_{P}(A_{1}:\cdots:A_{n}).
\end{equation}
\end{proof}
Now we give an upper bound on $\Delta_{P}$ in terms of a certain sum
of entanglement entropy. This is a generalization of property (IV). 
\begin{prop}
\label{prop:Delta upper bound}$\Delta_{P}$ is bounded from above
by 
\begin{equation}
\Delta_{P}(A_{1}:\cdots:A_{n})\leq\min_{i}[S_{A_{1}}+\cdots+S_{A_{1}\cdots A_{i-1}A_{i+1}\cdots A_{i}}+\cdots+S_{A_{n}}].
\end{equation}
\end{prop}

\begin{proof}
For simplicity, we will first prove this bound for tripartite state $\rho_{ABC}$
. Let us consider a standard purification of a given state
$\rho_{ABC}=\sum p_{k}\ket{\phi^{k}}\bra{\phi^{k}}_{ABC}$ such that
\begin{equation}\label{standP3}
\ket{\psi}_{AA'BB'CC'}=\sum_{k=1}^{{\rm rank}[\rho_{ABC}]}\sqrt{p_{k}}\ket{\phi^{k}}_{ABC}\otimes\ket{0}_{A'}\otimes\ket{0}_{B'}\otimes\ket{k}_{C'}.
\end{equation}
For this purification we have $\Delta_{P}(A:B:C)\leq S_{AA'}+S_{BB'}+S_{CC'}$. Note that $\rho_{AA'}=\rho_{A}\otimes\ket{0}\bra{0}_{A'}$
and the same for $B$, and $S_{CC'}=S_{AA'BB'}$, it can be easily
shown that 
\begin{equation}
S_{AA'}=S_{A},\ S_{BB'}=S_{B},\ S_{CC'}=S_{AB},
\end{equation}
for the state (\ref{standP3}). Thus we get $\Delta_{P}(A:B:C)\leq S_{A}+S_{B}+S_{AB}$.
Commuting $A,B,C$, we get three upper bounds on
$\Delta_{P}$, 
\begin{equation}
\Delta_{P}(A:B:C)\leq\min\{S_{A}+S_{B}+S_{AB},S_{B}+S_{C}+S_{BC},S_{C}+S_{A}+S_{CA}\}.
\end{equation}

The generalization of this proof to $n$-partite cases is straightforward. 
\end{proof}
These upper bounds indicates that if we have a bipartite state $\rho_{AB}$,
and consider extensions $\rho_{ABC}$, the upper bound of $\Delta_{P}(A:B:C)$
is totally determined by the information included in $\rho_{AB}$.
In other words, we can not arbitrarily increase the multipartite
correlations by adding ancillary systems (the upper bound can
be reached by any purification $\ket{\psi}_{ABC}$ of $\rho_{AB}$,
though it is not the only way to saturate these bounds as we
will see in corollary \ref{cor:Delta for SA saturation} and \ref{cor:Delta for SSA saturation}).\\

Next we state universal lower bounds on $\Delta_{P}$. We first show the following inequality satisfied for any tripartite state $\rho_{ABC}$. 
\begin{prop}
\label{prop:3-Delta lower bound}Tripartite entanglement of purification
$\Delta_{P}(A:B:C)$ is bounded from below by 
\begin{equation}
\Delta_{P}(A:B:C)\geq\max\{S_{A}+S_{B}+S_{C}-S_{ABC},2(S_{A}+S_{B}+S_{C})-S_{AB}-S_{BC}-S_{CA}\}.
\end{equation}
\end{prop}

\begin{proof}
Let us take an optimal purification $\ket{\psi}_{AA'BB''CC''}$. For
this state we have 
\begin{align}
\Delta_{P}(A:B:C) & =S_{AA'}+S_{BB'}+S_{CC'}\nonumber \\
 & =I(AA':BB')+I(AA'BB':CC')\nonumber \\
 & \geq I(A:B)+I(AB:C)\nonumber \\
 & =S_{A}+S_{B}+S_{C}-S_{ABC},
\end{align}
where in the third line we used the monotonicity of mutual information
\begin{equation}
I(AX:B)\geq I(A:B).
\end{equation}
Moreover, for tripartite pure states $\ket{\psi}_{AA'BB'CC'}$, we
have $I(AA'BB':CC')=I(AA':CC')+I(BB':CC')$. Thus we get 
\begin{align}
\Delta_{P}(A:B:C) & =I(AA':BB')+I(BB':CC')+I(CC':AA')\nonumber \\
 & \geq I(A:B)+I(B:C)+I(C:A)\nonumber \\
 & =2(S_{A}+S_{B}+S_{C})-S_{AB}-S_{BC}-S_{CA}.
\end{align}
\end{proof}
These bounds have a suggestive form in terms of tripartite mutual information\footnote{Note that it leads to a general relationship $\Delta_P(A:B:C)\geq\tilde{I_3}(A:B:C)$ for any quantum states.}:
\begin{equation}
\Delta_{P}(A:B:C)\geq\max\{I(A:B:C),\ I(A:B:C)+\tilde{I}_{3}(A:B:C)\}.
\end{equation}
In particular, $\Delta_{P}$ is always greater or equal to the multipartite
mutual information $I(A:B:C)$. We will see that this is also
true for $n$-partite states in the following proposition. 
\begin{prop}
\label{prop:n-partite Delta > MI}The multipartite entanglement of
purification $\Delta_{P}(A_{1}:\cdots:A_{n})$ is bounded from below
by the multipartite mutual information, 
\begin{equation}
\Delta_{P}(A_{1}:\cdots:A_{n})\geq I(A_{1}:\cdots:A_{n}).
\end{equation}
\end{prop}

\begin{proof}
The proof is essentially the same as that in tripartite case. Let us consider an optimal
purification $\ket{\psi}_{A_{1}A_{1}'\cdots A_{n}A_{n}'}$ for $\rho_{A_{1}\cdots A_{n}}$.
Then we have 
\begin{align}
\Delta_{P}(A_{1}:\cdots:A_{n}) & =\sum_{i=1}^{n}S_{A_{i}A_{i}'}=\sum_{i=1}^{n}S_{A_{i}A_{i}'}-S_{A_{1}A_{1}'\cdots A_{n}A_{n}'}\nonumber \\
 & =I(A_{1}A_{1}':\cdots:A_{n}A_{n}')\nonumber \\
 & \geq I(A_{1}:\cdots:A_{n}),
\end{align}
where we used the property of multipartite mutual information \citep{YHHHOS:mEsq}
\begin{equation}
I(A_{1}X:\cdots:A_{n})\geq I(A_{1}:\cdots:A_{n}).
\end{equation}

\end{proof}
This is a generalization of property (Va), (Vb) and provides general
relationship between two types of multipartite total correlation measures
$I$ and $\Delta_{P}$. It is worth to point out that these two quantities behave very similarly.
Indeed, the propositions \ref{lem:Delta decoupling relations}, \ref{prop:Delta for pure},
\ref{prop:Fauthfullness of Delta }, \ref{prop:Delta monotonicity},
and \ref{prop:Delta upper bound} are also true for multipartite mutual
information. One exception is, the lower bound of
tripartite case, $\Delta_{P}(A:B:C)\geq I(A:B:C)+\tilde{I}_{3}(A:B:C)$, which
is obviously violated for $I(A:B:C)$ when $\tilde{I}_{3}(A:B:C)$
is positive.\\

The proposition \ref{prop:n-partite Delta > MI} also allows us to
give a simple proof of the proposition \ref{prop:Fauthfullness of Delta }. 
\begin{proof}
If a $n$-partite state is totally product $\rho_{A_{1}\cdots A_{n}}=\rho_{A_{1}}\otimes\cdots\otimes\rho_{A_{n}}$,
then one can get $\Delta_{P}=0$ by purifying each subsystems independently.
On the other hand, if $\Delta_{P}(A_{1}:\cdots:A_{n})=0$, then $I(A_{1}:\cdots:A_{n})=S(\rho_{A_{1}\cdots A_{n}}||\rho_{A_{1}}\otimes\cdots\otimes\rho_{A_{n}})=0$
following the proposition \ref{prop:n-partite Delta > MI}. Thus, the non-degeneracy
of relative entropy leads to $\rho_{A_{1}\cdots A_{n}}=\rho_{A_{1}}\otimes\cdots\otimes\rho_{A_{n}}$. 
\end{proof}
Using the above arguments, some properties of $\Delta_{P}$ follow
as corollaries. 
\begin{cor}
For any pure $n$-partite state, $\Delta_{P}$ is polygamous:
\begin{equation}
\Delta_{P}(A_{1}:\cdots:A_{n-1}:BC)\leq\Delta_{P}(A_{1}:\cdots:A_{n-1}:B)+\Delta_{P}(A_{1}:\cdots:A_{n-1}:C).
\end{equation}
\end{cor}

\begin{proof}
For any pure state $\ket{\phi}_{A_{1}\cdots A_{n}BC}$ it follows
from the proposition \ref{prop:Delta for pure} that 
\begin{align}
\Delta_{P}(A_{1}:\cdots:A_{n-1}:BC) & =\sum_{i=1}^{n-1}S_{A_{i}}+S_{BC}=\sum_{i=1}^{n-1}S_{A_{i}}+S_{A_{1}\cdots A_{n-1}}\nonumber \\
 & \leq2\sum_{i=1}^{n-1}S_{A_{i}}\nonumber \\
 & =\sum_{i=1}^{n-1}S_{A_{i}}+S_{B}-S_{A_{1}\cdots A_{n-1}B}+\sum_{i=1}^{n-1}S_{A_{i}}+S_{C}-S_{A_{1}\cdots A_{n-1}C}\nonumber \\
 & \leq\Delta(A_{1}:\cdots:A_{n-1}:B)+\Delta(A_{1}:\cdots:A_{n}:C),
\end{align}
where in the first inequality we used the subadditivity of von Neumann
entropy recursively, in the third line $S_{B}=S_{A_{1}\cdots A_{n-1}C},\ S_{C}=S_{A_{1}\cdots A_{n-1}B}$
for pure states, and in the last inequality the proposition \ref{prop:n-partite Delta > MI}. 
\end{proof}
As $E_{P}$ is so, the multipartite entanglement of purification
is difficult to calculate in general because of the minimization
over infinitely many purifications. Nevertheless, there is a class of
quantum states for which one can rigorously calculate $\Delta_{P}$
in terms of entanglement entropy.
\begin{cor}
\label{cor:Delta for SA saturation}For a class of tripartite states
$\rho_{ABC}$ that saturate the subadditivity $S_{ABC}=S_{C}-S_{AB}$,
we have $\Delta_{P}(A:B:C)=S_{A}+S_{B}+S_{AB}$. 
\end{cor}

\begin{proof}
From proposition \ref{prop:Delta upper bound} and
\ref{prop:3-Delta lower bound}, we have 
\begin{equation}
S_{A}+S_{B}+S_{C}-S_{ABC}\leq\Delta_{P}(A:B:C)\leq S_{A}+S_{B}+S_{AB}.
\end{equation}
Therefore $S_{C}-S_{ABC}=S_{AB}$ leads $\Delta_{P}(A:B:C)=S_{A}+S_{B}+S_{AB}$. 
\end{proof}
\begin{cor}
\label{cor:Delta for SSA saturation}For a class of tripartite states
$\rho_{ABC}$ that saturate both of the two forms of the strong subadditivity,
$S_{A}+S_{C}=S_{AB}+S_{BC}$ and $S_{B}+S_{C}=S_{AB}+S_{AC}$, then
we have $\Delta_{P}(A:B:C)=S_{A}+S_{B}+S_{AB}$. 
\end{cor}

\begin{proof}
From proposition \ref{prop:Delta upper bound} and
\ref{prop:3-Delta lower bound}, we have 
\begin{equation}
2(S_{A}+S_{B}+S_{C})-S_{AB}-S_{BC}-S_{CA}\leq\Delta_{P}(A:B:C)\leq S_{A}+S_{B}+S_{AB},
\end{equation}
where the lower bound can be expressed as
\begin{align}
 & 2(S_{A}+S_{B}+S_{C})-S_{AB}-S_{BC}-S_{CA}\nonumber \\
= & S_{A}+S_{B}+S_{AB}+(S_{A}+S_{C}-S_{AB}-S_{BC})+(S_{B}+S_{C}-S_{AB}-S_{AC}).
\end{align}
Thus if the two strong subadditivity are simultaneously saturated,
we get $\Delta_{P}(A:B:C)=S_{A}+S_{B}+S_{AB}$.
\end{proof}
We also provide another lower bound on $\Delta_{P}$ in terms
of bipartite entanglement of purification. 
\begin{prop}\label{prop:Delta>3Ep}
$\Delta_{P}$ is bounded from below by 
\begin{equation}
\Delta_{P}(A_{1}:\cdots:A_{n})\geq\sum_{i=1}^{n}E_{P}(A_{i}:A_{1}\cdots A_{i-1}A_{i+1}\cdots A_{n}).
\end{equation}
\end{prop}

\begin{proof}
Let us prove it for tripartite cases for simplicity. For a state $\rho_{ABC}$,
we have 
\begin{align}
\Delta_{P}(A:B:C) & =\min_{\ket{\psi}_{AA'BB'CC'}}[S_{AA'}+S_{BB'}+S_{CC'}]\nonumber \\
 & \geq\min_{\ket{\psi}_{AA'BB'CC'}}S_{AA'}+\min_{\ket{\psi}_{AA'BB'CC'}}S_{BB'}+\min_{\ket{\psi}_{AA'BB'CC'}}S_{CC'}\nonumber \\
 & =E_{P}(A:BC)+E_{P}(B:CA)+E_{P}(C:AB),
\end{align}
then the bound follows.
\end{proof}

\section{Holography: multipartite entanglement wedge cross section\label{sec3}}

In this section, we define a multipartite generalization of entanglement
wedge cross-section introduced in \citep{TU:HEoP,NDHZS:HEoP}, motivated
by the tensor network description of AdS geometry \citep{Swingle:TensorNetwork,MiyajiTakayanagi:SurfaceStateCorrespondence,MTW:OptimizationofPathIntegral1,CKMTW:OptimizationofPathIntegral2}. 

\subsection{Definition}

We start by setting our conventions in holography. To compute entanglement entropy in quantum field theories, we often choose a (maybe disconnected) subsystem $A$ on a time slice, and
the Hilbert space of the field theory is factorized into ${\cal H}_{tot}={\cal H}_{A}\otimes{\cal H}_{A^{c}}$.
Then the entanglement entropy for subsystem $A$ in a state $\rho_{tot}$,
is defined as the von Neumann entropy of the reduced density matrix
$\rho_{A}$, 
\begin{equation}
S_{A}:=-\text{Tr}\rho_{A}\log\rho_{A}\ ,\quad\rho_{A}=\text{Tr}_{A^{c}}\rho_{tot}\ .
\end{equation}
In the AdS/CFT correspondence, the holographic entanglement entropy
formula \citep{RT:RT-formula,HRT:HRT-formula} tells us how to calculate
entanglement entropy in dual gravity side. Consider a state in
$d$-dimensional holographic CFT which has a classical $d+1$ dimensional
gravity dual. In the present paper, we will restrict ourselves to static cases\footnote{Even though we can define a covariant version of multipartite entanglement
wedge cross-section in a similar way for $E_{W}$ \citep{TU:HEoP},
the holographic proofs of its entropic inequalities for more than
4-partite cases is not straightforward \citep{RW:MaximinisNotEnough}.
We leave it as a future problem.}. The dual gravity solution of the given state $\rho_{tot}$ will
be a canonical time slice $M$ in gravity side. To compute the entanglement
entropy for a chosen subsystem $A\subset\partial M$, we are looking
for a $d-1$-dimensional surface $\Gamma_{A}$ in $M$ with minimal
area, under the conditions that $\partial\Gamma_{A}=\partial A$ and $\Gamma_{A}$
is homologous to $A$. Then the holographic entanglement entropy is determined
by the minimal area\footnote{We work at the leading order of large $N$ limit
through the whole of present paper.},
\begin{equation}
S_{A}=\frac{\text{Area}(\Gamma_{A}^{min})}{4G_{N}}\ .\label{EE}
\end{equation}

Let us start to define our main interest, i.e. a multipartite generalization
of entanglement wedge cross-section $\Delta_{W}$. We mostly focus
on tripartite case for simplicity, though the generalization to more
partite cases is rather straightforward. We take subsystems $A$,
$B$ and $C$ on the boundary $\partial M$. In general $\rho_{ABC}$
is a mixed state. Then one can compute the holographic entanglement
entropy $S_{A},S_{B},S_{C}$ and also $S_{ABC}$ following (\ref{EE}).
The corresponding minimal surfaces are denoted by $\Gamma_{A}^{min}$,
$\Gamma_{B}^{min}$, $\Gamma_{C}^{min}$, $\Gamma_{ABC}^{min}$, respectively.
The entanglement wedge $M_{ABC}$ \citep{CKNR:EntanglementWedge,Wall:Maximinsurfaces,HHLR:EntanglementWedge}
is defined as a region of $M$ \footnote{More precisely, we consider a constant time slice of entanglement
wedge and call it also entanglement wedge, while the former is codimension-0
and the latter is codimension-1.} with boundary $A,B,C$ and $\Gamma_{ABC}^{min}$: 
\begin{equation}
\partial M_{ABC}=A\cup B\cup C\cup\Gamma_{ABC}^{min}\ .
\end{equation}
Notice that $M_{ABC}$ gets disconnected when some of $A,B,C$ or
all of them are decoupled. Also note that $\partial M_{ABC}$ may
include bifurcation surfaces in the bulk such as in AdS black hole
geometry.

Next, we divide arbitrarily the boundary $\partial M_{ABC}$ (not
$M_{ABC}$ itself) into three parts $\widetilde{A},\widetilde{B},\widetilde{C}$
so that they satisfy 
\begin{equation}
\widetilde{A}\cup\widetilde{B}\cup\widetilde{C}=\partial M_{ABC}\ ,\label{constraint1}
\end{equation}
and 
\begin{equation}
A\subset\widetilde{A}\ ,B\subset\widetilde{B}\ ,C\subset\widetilde{C}\ .\label{constraint2}
\end{equation}
The boundary of $\widetilde{A},\widetilde{B},\widetilde{C}$ is denoted by ${\cal D}_{ABC}$. Regarding each $\widetilde{A},\widetilde{B},\widetilde{C}$ as a subsystem
of a geometric pure state, we can calculate 
\begin{equation}
S_{\widetilde{A}}+S_{\widetilde{B}}+S_{\widetilde{C}},\label{sum}
\end{equation}
by using holographic entanglement entropy formula (\ref{EE}). This
is performed by finding a minimal surface $\Sigma_{ABC}^{min}$ that
consists of three parts $\Sigma_{A}$, $\Sigma_{B}$, $\Sigma_{C}$,
which share the boundary ${\cal D}_{ABC}$, such that 
\begin{equation}
\Sigma_{ABC}^{min}=\Sigma_{A}\cup\Sigma_{B}\cup\Sigma_{C}\ ,\quad\partial\Sigma_{ABC}^{min}={\cal D}_{ABC}\ ,\label{bounded}
\end{equation}
and 
\begin{equation}
\Sigma_{A,B,C}~~\text{is homologous to}~~\widetilde{A},\widetilde{B},\widetilde{C}~~\text{inside}~~M_{ABC}\ .
\end{equation}
Since $\partial M_{ABC}$ is codimention-2, the surfaces ${\cal D}_{ABC}$
which plays the role of the division of $\partial M_{ABC}=\widetilde{A}\cup\widetilde{B}\cup\widetilde{C}$,
is codimension-3. In the case of AdS$_{3}$/CFT$_{2}$, ${\cal D}_{ABC}$
is in general three separated points on $\Gamma_{ABC}^{min}$, see
Fig \ref{fig:MEoP2}, \ref{fig:MEoP2BH}.

\begin{figure}[H]
\centering{}\includegraphics[scale=0.33]{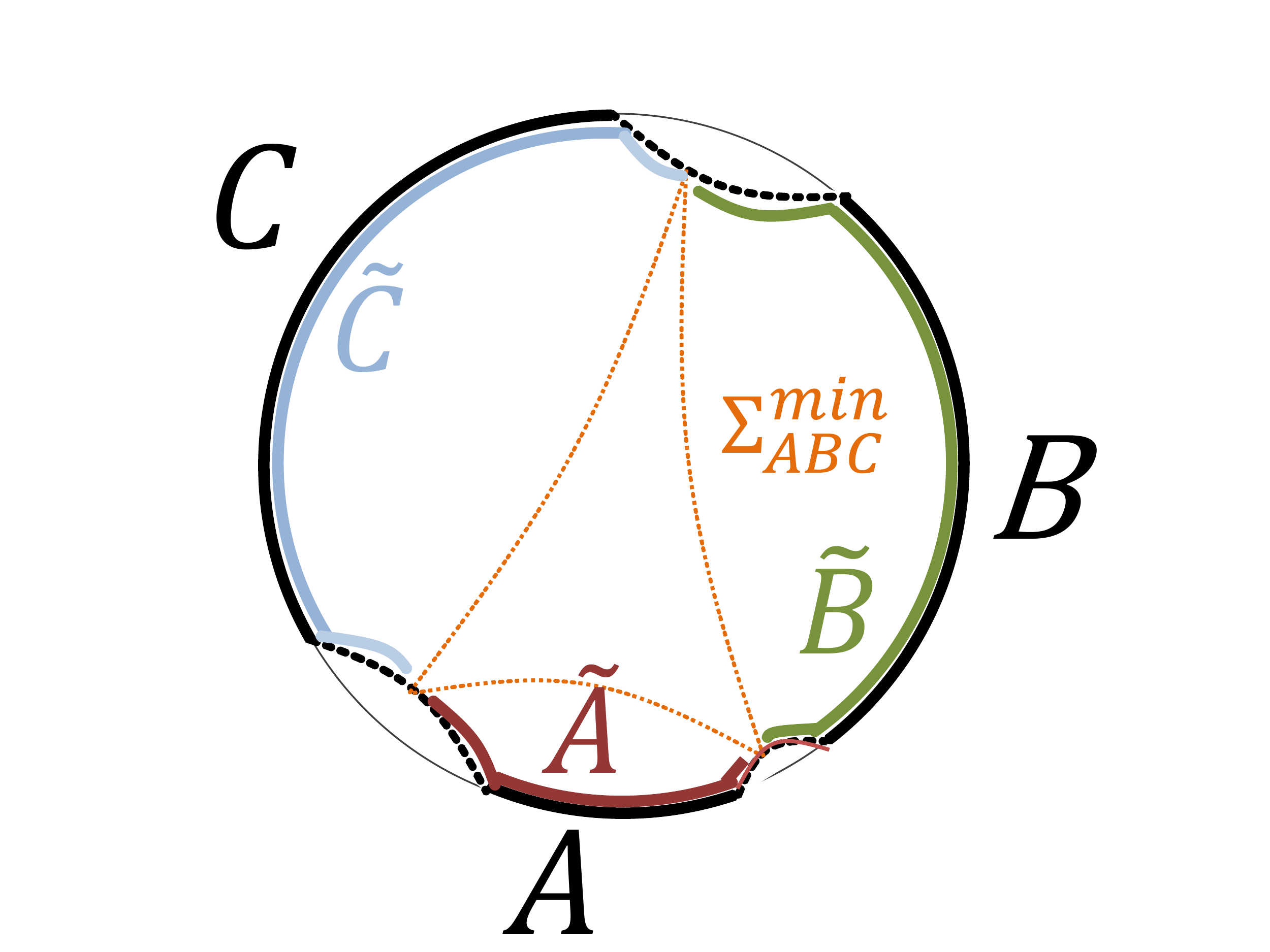}\caption{\label{fig:MEoP2} An example of tripartite entanglement wedge cross-section.
The black bold dashed lines represents the minimal surface $\Gamma_{ABC}^{min}$,
giving a part of the boundary of $M_{ABC}$. The yellow thin dashed lines
represents $\Sigma_{ABC}^{min}$ whose area (divided by 4$G_{N}$)
is $\Delta_{W}$.}
\end{figure}
\begin{figure}[H]
\centering{}\includegraphics[scale=0.33]{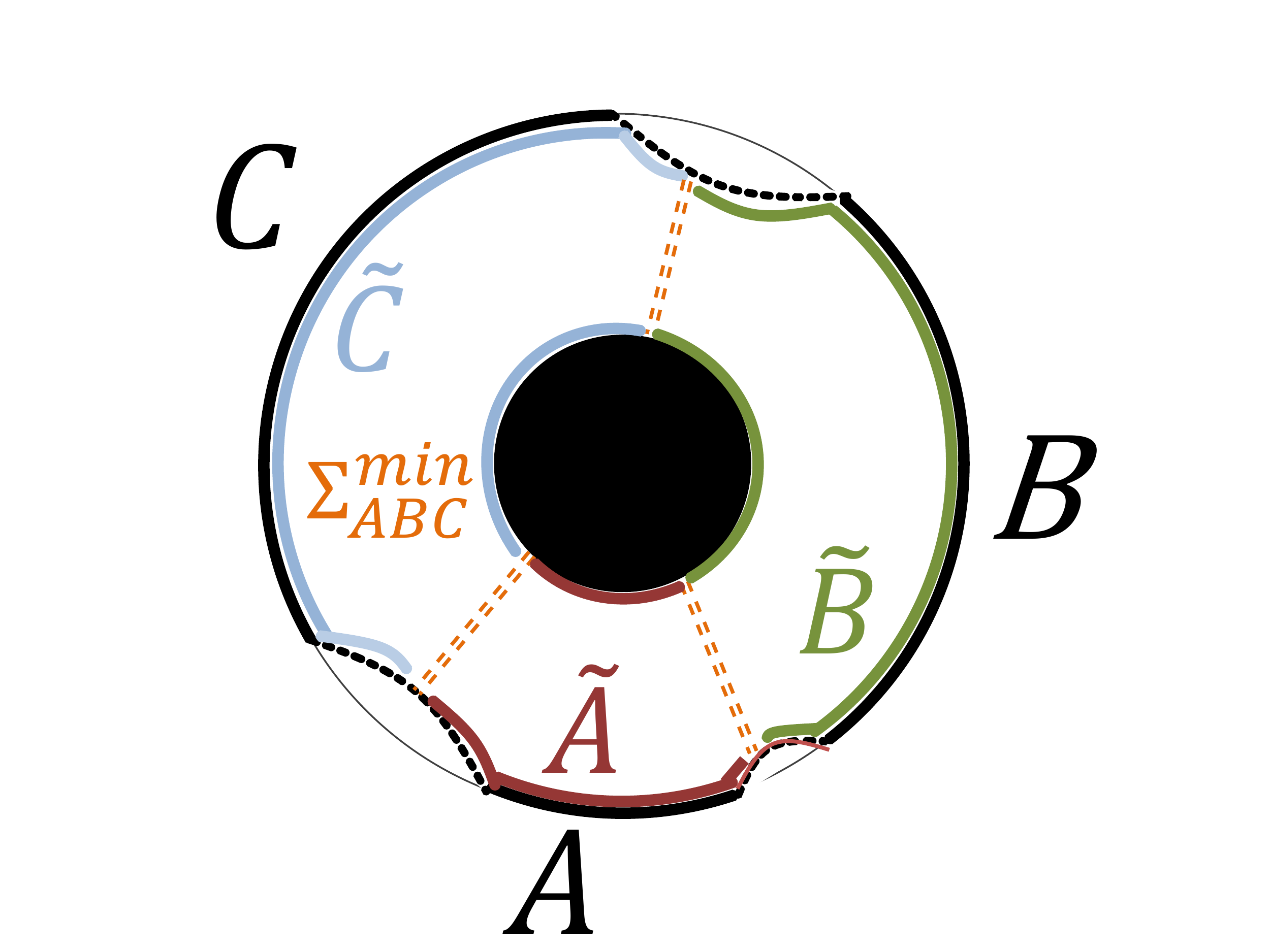}\caption{\label{fig:MEoP2BH} An example of tripartite entanglement wedge cross-section
in a black hole geometry. Each surface of $\Sigma_{ABC}^{min}$
is doubled.}
\end{figure}

Finally we minimize the area of $\Sigma_{ABC}^{min}$ over all possible
divisions $\widetilde{A},\widetilde{B},\widetilde{C}$ that satisfy
the conditions (\ref{constraint1}) and (\ref{constraint2}). This
gives now a quantity we call the multipartite entanglement wedge cross
section 
\begin{equation}
\Delta_{W}(\rho_{ABC}):=\min_{\widetilde{A},\widetilde{B},\widetilde{C}}\left[\frac{\text{Area}(\Sigma_{ABC}^{min})}{4G_{N}}\right]\ .\label{defDeltaW}
\end{equation}
For $n$-partite boundary subsystems, in general, the multipartite
entanglement wedge cross-section is defined in the same manner. Notice
that when it is reduced to the bipartite case, this definition is
actually twice of the bipartite entanglement wedge cross-section defined
in \citep{TU:HEoP,NDHZS:HEoP}. We sometimes write $\Delta_{W}(\rho_{ABC})=\Delta_{W}(A:B:C)$
to clarify a way of partition of subsystems.

In summary, $\Delta_{W}$ computes the multipartite cross-sections
of the entanglement wedge $M_{ABC}$ and it is a natural generalization
of the bipartite entanglement wedge cross-section. This can be used
as a total measure of how strongly multiple parties are holographically
connected with each other. Below we study the properties of $\Delta_{W}$.

\subsection{Properties of multipartite entanglement wedge cross-section and
the conjecture $\Delta_{W}=\Delta_{P}$}

In the following, we investigate holographic properties of $\Delta_{W}$,
inspired by those of $\Delta_{P}$. To avoid unnecessary complexity,
we mostly consider tripartite case only, and it should be understood that the properties are easily generalized for $n$-partite cases in somewhat trivial ways unless otherwise emphasized. \\

First, if $\rho_{ABC}$ is pure, from the definition (\ref{bounded}),
$\Sigma_{ABC}^{min}$ coincides with $\Gamma_{A}^{min}\cup\Gamma_{B}^{min}\cup\Gamma_{C}^{min}$.
Therefore $\Delta_{W}$ is equal to the sum of the entanglement entropy
of $A$, $B$ and $C$: 
\begin{equation}
\Delta_{W}(A:B:C)=S_{A}+S_{B}+S_{C}.
\end{equation}

As we mentioned above, for a partly decoupled entanglement wedges
i.e. if $M_{ABC}=M_{AB}\bigsqcup M_{C}$, where $\bigsqcup$ denote
that the geometries $M_{AB}, M_{C}$ are totally separated,
then $\Delta_{W}$ is reduced to (twice of) entanglement wedge
cross-section. This clearly leads that $\Delta_{W}=0$ if
and only if the entire entanglement wedge is totally decoupled $M_{A_{1}\cdots A_{n}}=\bigsqcup_{i=1}^{n}M_{A_{n}}$
for multipartite setups.\\

One can easily show that $\Delta_{W}$ decreases when we reduce one
of the subregions in $A,B,C\equiv C_{1}\cup C_{2}$: 
\begin{equation}
\Delta_{W}(A:B:C_{1}\cup C_{2})\geq\Delta_{W}(A:B:C_{1})\ ,
\end{equation}
by using the so-called entanglement wedge nesting, 
\begin{equation}
M_{X}\subset M_{XY},
\end{equation}
which holds for any boundary subregions $X,\ Y$ \citep{CKNR:EntanglementWedge,Wall:Maximinsurfaces,HHLR:EntanglementWedge}.\\

One can also easily show an upper bound of $\Delta_{W}$ by a graph
proof (Fig.\ref{fig:MEoP6}): 
\begin{equation}
\Delta_{W}(A:B:C)\leq S_{A}+S_{B}+S_{AB}\ ,
\end{equation}
By commuting $A,B,C$, one can further get 
\begin{equation}
\Delta_{W}(A:B:C)\leq\min[S_{A}+S_{B}+S_{AB},S_{B}+S_{C}+S_{BC},S_{A}+S_{C}+S_{AC}].
\end{equation}

\begin{figure}[H]
\centering{}\includegraphics[scale=0.33]{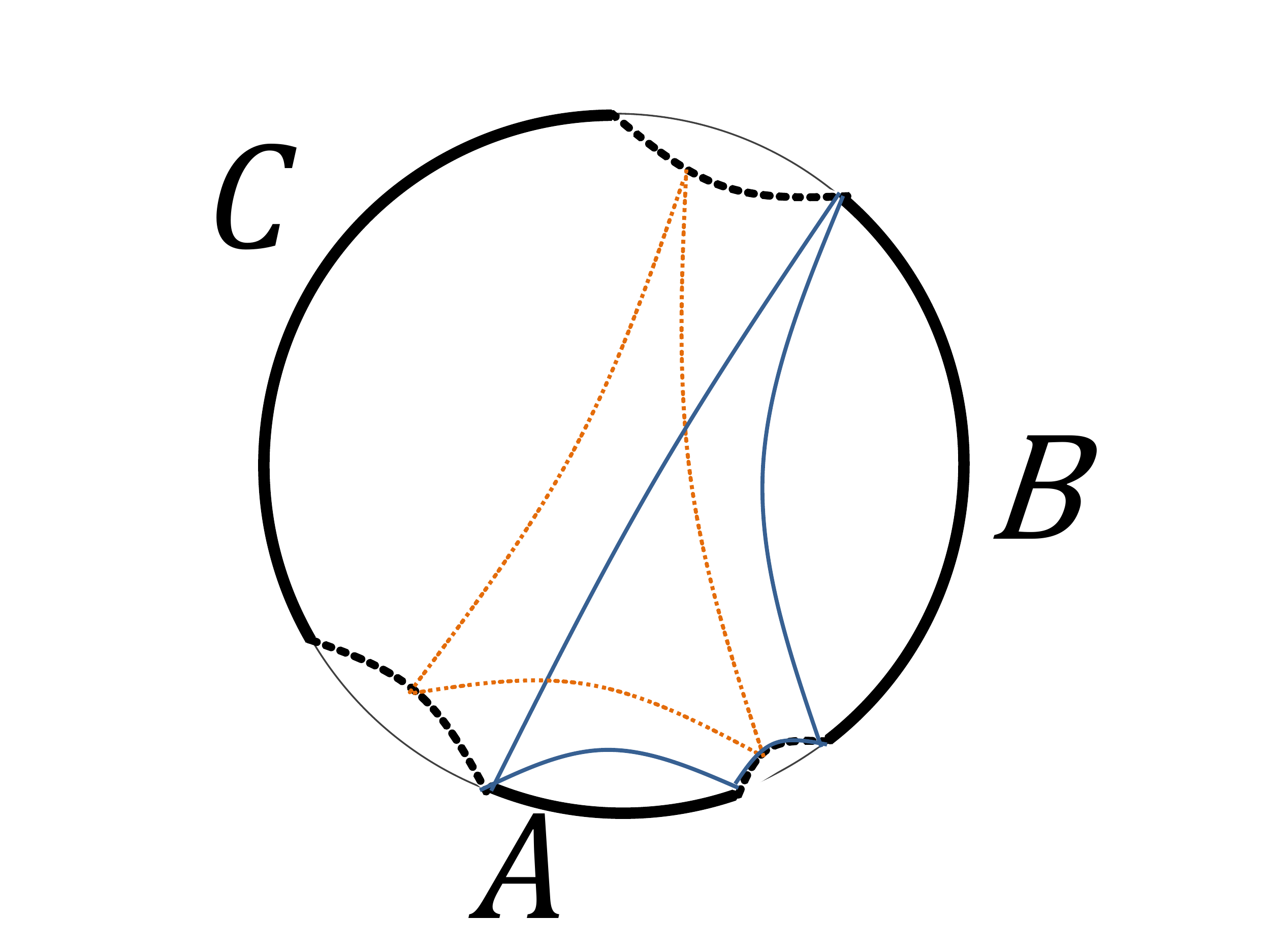}\caption{\label{fig:MEoP6} The proof of an upper bound of $\Delta_{W}$. The
sum of blue real lines is $S_{A}+S_{B}+S_{AB}$ and the sum of dashed
yellow lines are $\Delta_{W}(\rho_{ABC})$. Clearly $\Delta_{W}(\rho_{ABC})\protect\leq S_{A}+S_{B}+S_{AB}$ holds somewhat trivially since $S_{A}+S_{B}+S_{AB}$ has UV divergences
while $\Delta_{W}$ does not. When two of subsystems share the boundary,
$\Delta_{W}$ also diverges, but it is always weaker than $S_{A}+S_{B}+S_{AB}$
shown by a graph. }
\end{figure}

Furthermore, for tripartite setups, one can show two lower bounds
of $\Delta_{W}$ also by graph proofs (Fig.\ref{fig:MEoP3}, Fig.\ref{fig:MEoP3Y}): 
\begin{equation}
\Delta_{W}(A:B:C)\geq I(A:B:C)=S_{A}+S_{B}+S_{C}-S_{ABC}\ .
\end{equation}
\begin{equation}
\Delta_{W}(A:B:C)\geq I(A:B:C)+\tilde{I}_{3}(A:B:C)=2(S_{A}+S_{B}+S_{C})-S_{AB}-S_{BC}-S_{AC}\ .
\end{equation}
Note that, however, in holography we always have $I(A:B:C)\geq I(A:B:C)+\tilde{I}_{3}(A:B:C)$
because of negative tripartite information $\tilde{I}_{3}\leq0$ \citep{HHM:MonogamyofHMI}.
Therefore the former is always tighter.

Similarly, for $n$-partite setup, one can easily show that 
\begin{equation}
\Delta_{W}(A_{1}:\cdots:A_{n})\geq I(A_{1}:\cdots:A_{n}),
\end{equation}
by drawing graphs.

\begin{figure}[H]
\centering{}\includegraphics[scale=0.33]{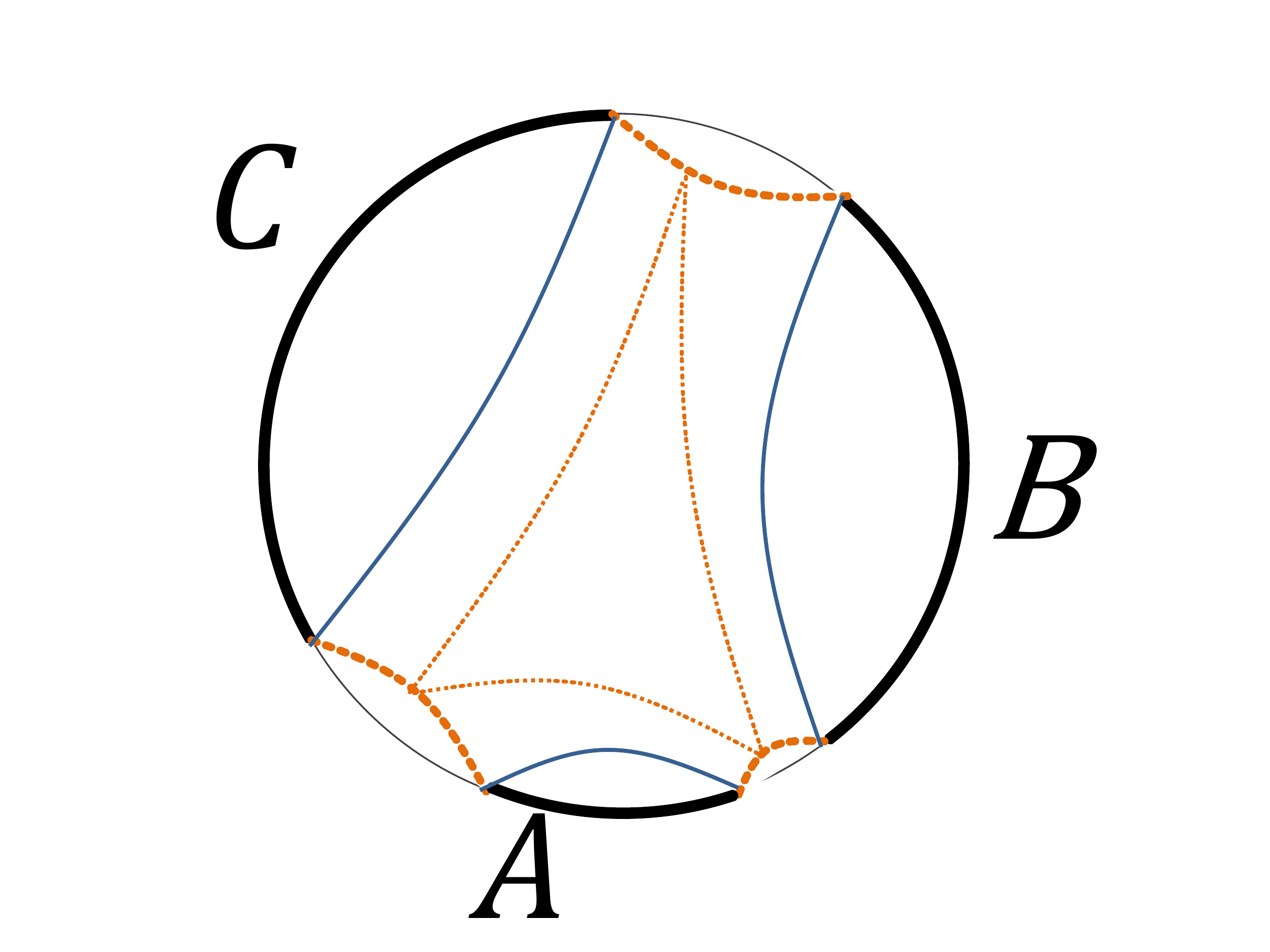}\caption{\label{fig:MEoP3} The proof of a lower bound of $\Delta_{W}$. The
sum of dashed yellow lines is $\Delta_{ABC}+S_{ABC}$ and the sum
of real blue lines are $S_{A}+S_{B}+S_{C}$. Clearly $\Delta_{W}(\rho_{ABC})+S_{ABC}\protect\geq S_{A}+S_{B}+S_{C}$
follows since the entanglement entropy are defined as minimal surfaces. }
\end{figure}
\begin{figure}[H]
\centering{}\includegraphics[scale=0.33]{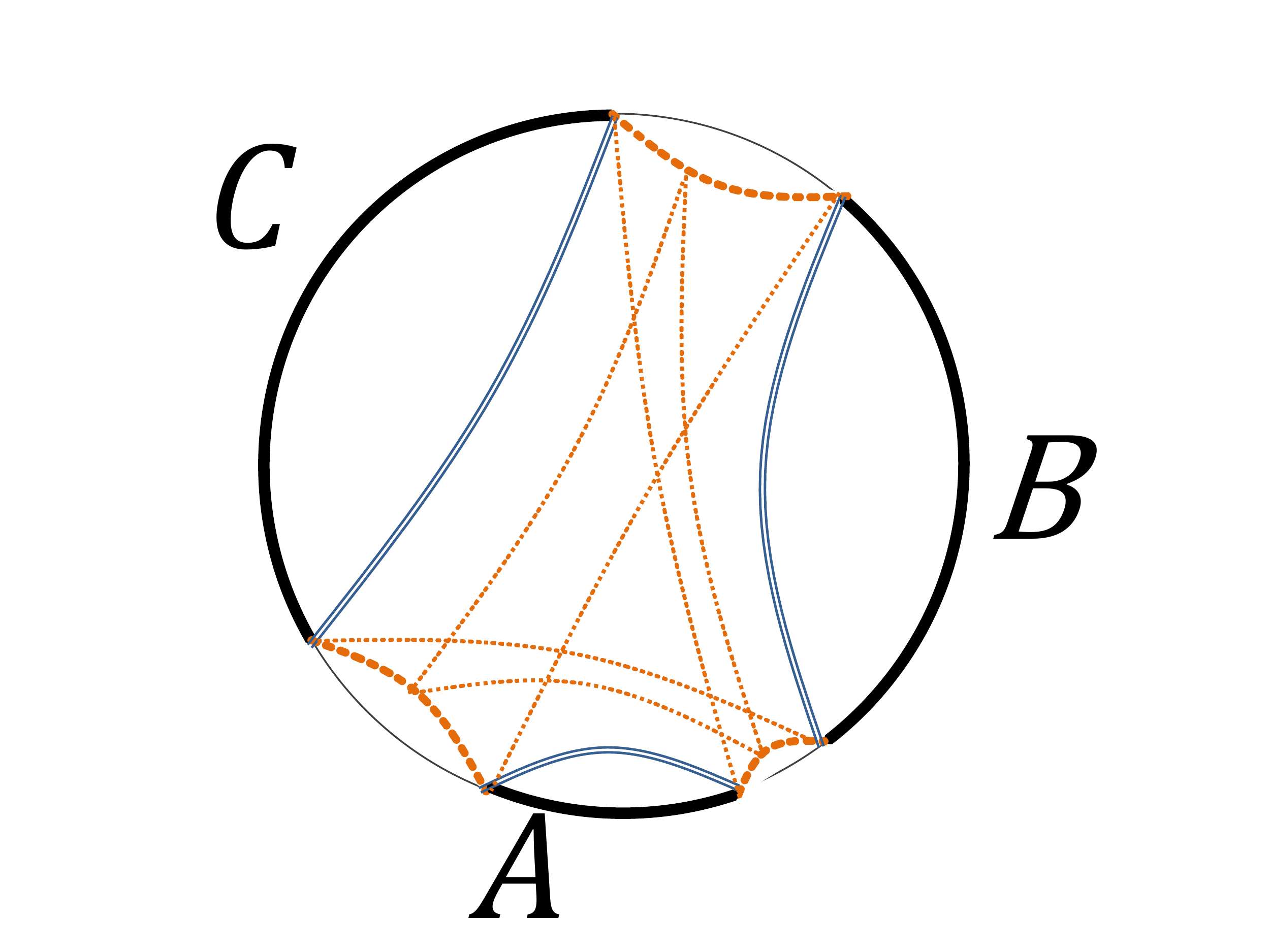}\caption{\label{fig:MEoP3Y} The proof of a lower bound of $\Delta_{W}$. The
sum of dashed yellow lines is $\Delta_{ABC}+S_{AB}+S_{BC}+S_{CA}$,
and the sum of real doubled blue lines are $2(S_{A}+S_{B}+S_{C})$.
Clearly $\Delta_{W}+S_{AB}+S_{BC}+S_{CA}\protect\geq 2(S_{A}+S_{B}+S_{C})$
follows, since the entanglement entropy are defined as minimal surfaces. }
\end{figure}

Note that the corollaries in the previous section automatically follows
for $\Delta_{W}$ from the above discussion, while we can also show them by drawing
graphs. We also note that the proposition \ref{prop:Delta>3Ep} for $\Delta_W$ can be easily shown by graphs. \\

One can view the above properties are the multipartite generalization
of the holographic properties of bipartite entanglement wedge cross-section.
Motivated by the same properties of $\Delta_{W}$ and $\Delta_{P}$,
we now make a conjecture that, the multipartite entanglement wedge
cross-section $\Delta_{W}$ we defined in this section is nothing
but the holographic counterpart of the multipartite entanglement of
purification $\Delta_{P}$ we defined in the last section (at the leading order $O(N^{2})$): 
\begin{equation}
\Delta_{W}=\Delta_{P}\ .\label{conjecture}
\end{equation}

\subsection{Computation of $\Delta_{W}$ in pure AdS$_{3}$}

Now we compute $\Delta_{W}$ in the simple examples of AdS$_{3}$/CFT$_{2}$.
We work in Poincar\'e patch, and a static ground state of a CFT$_{2}$
on a infinite line is described by a bulk solution with the metric
\begin{equation}
ds^{2}=\frac{dx^{2}+dz^{2}}{z^{2}}\ ,\quad x\in(-\infty,+\infty),z\in[0,+\infty)\ .
\end{equation}
The three subsystems we choose are the intervals $A=[-b,-a-r]$, $B=[-a+r,a-r]$,
$C=[a+r,b]$, where $b>a>0$ and $r$ is relatively small compared
to both $a$ and $b$. We require that the entanglement wedge of $ABC$
is connected, as shown in Fig.\ref{fig:MEoP4}. Following the definition
of $\Delta_{W}$ given in (\ref{defDeltaW}), in this set up the problem
becomes to find a triangle type configuration with the minimal length of three connected geodesics, where the ending points
of the geodesics are located on 3 semi-circles separately, as shown
in Fig.\ref{fig:MEoP4}. Since we focus on the case of 3 intervals
$A,B,C$ which have a reflection symmetry $x\to-x$, the reasonable
minimal configuration should also keep the reflection symmetry. This
consideration reduces the problem to find a special angle $\theta$
such that the length of 3 geodesics is minimal. Then, the tripartite
entanglement wedge cross-section $\Delta_{W}$ is given by 
\begin{equation}
\Delta_{W}(A:B:C)=\min_{\theta}\left[\frac{L(\theta)}{4G_{N}}\right]\ .
\end{equation}

We obtained a compact formula of the length $L$ as a function of
$a,b,r$ and $\theta$, however this formula is rather complicated.
We instead show numerical $\theta$ dependence of $L$ as plotted in
Fig.\ref{fig:Plot1} and evaluate the special values of both $\theta$
and $L$ satisfying the minimal length condition for a given $a,b,r$.

It is also straightforward to check the properties of $\Delta_{W}$
studied before in this particular setup.

\begin{figure}[H]
\centering{}\includegraphics[scale=0.33]{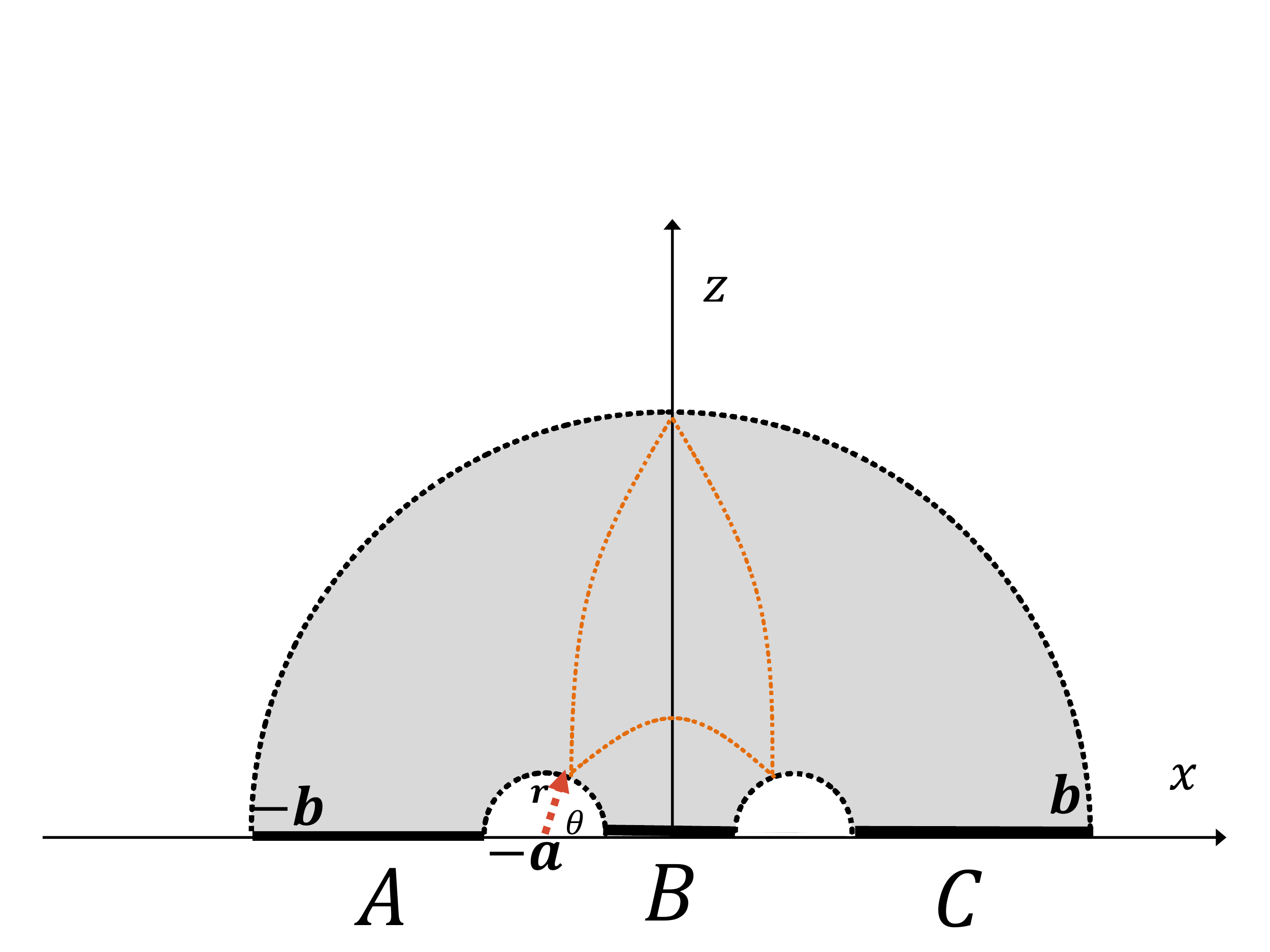}\caption{\label{fig:MEoP4} The computation of $\Delta_{W}$ in pure AdS$_{3}$. }
\end{figure}

\begin{figure}[H]
\centering{}\includegraphics[scale=0.6]{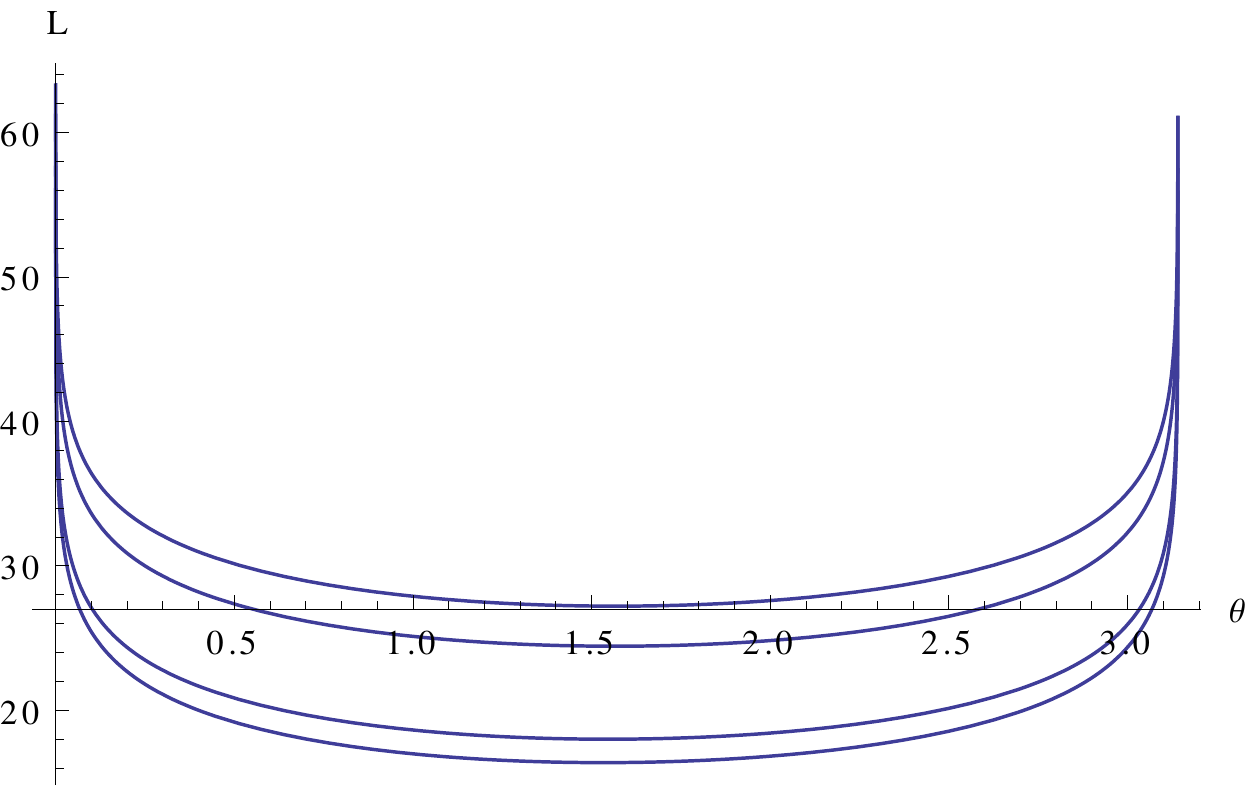}\caption{\label{fig:Plot1} $L-\theta$ plot in the computation of $\Delta_{W}$
with different parameters. From top to bottom, the parameters of $(a,b,r)$
are $(10,100,0.05),(10,100,0.1),(10,100,0.5),(10,100,0.75)$. $\Delta_{W}*4G_N$
and the the optimal value of $\theta$ are $(27.2046,\theta \to 1.56818)$ $(24.432,\theta\to1.5657),(17.9942,\theta\to1.54531),(16.3723,\theta\to1.53255)$,
respectively.}
\end{figure}

\subsection{Computation of $\Delta_{W}$ in BTZ black hole}

Now we turn to the BTZ black holes. A planar BTZ black hole describes
a $2$d CFT on a infinite line at finite temperature. The metric of
a fixed time slice of BTZ is given by 
\begin{equation}
ds^{2}=\frac{1}{z^{2}}\left(\frac{dz^{2}}{f(z)}+dx^{2}\right)\ ,\quad f(z)=1-\frac{z^{2}}{z_{H}^{2}}\ ,
\end{equation}
where the temperature is related to the horizon by $\beta=2\pi z_{H}$.
For simplicity we choose 3 subsystems $A$, $B$, $C$ as intervals
$[-\ell,0]$, $[0,\ell]$ and the remaining part of the infinite line,
respectively.

\begin{figure}[H]
\centering{}\includegraphics[scale=0.33]{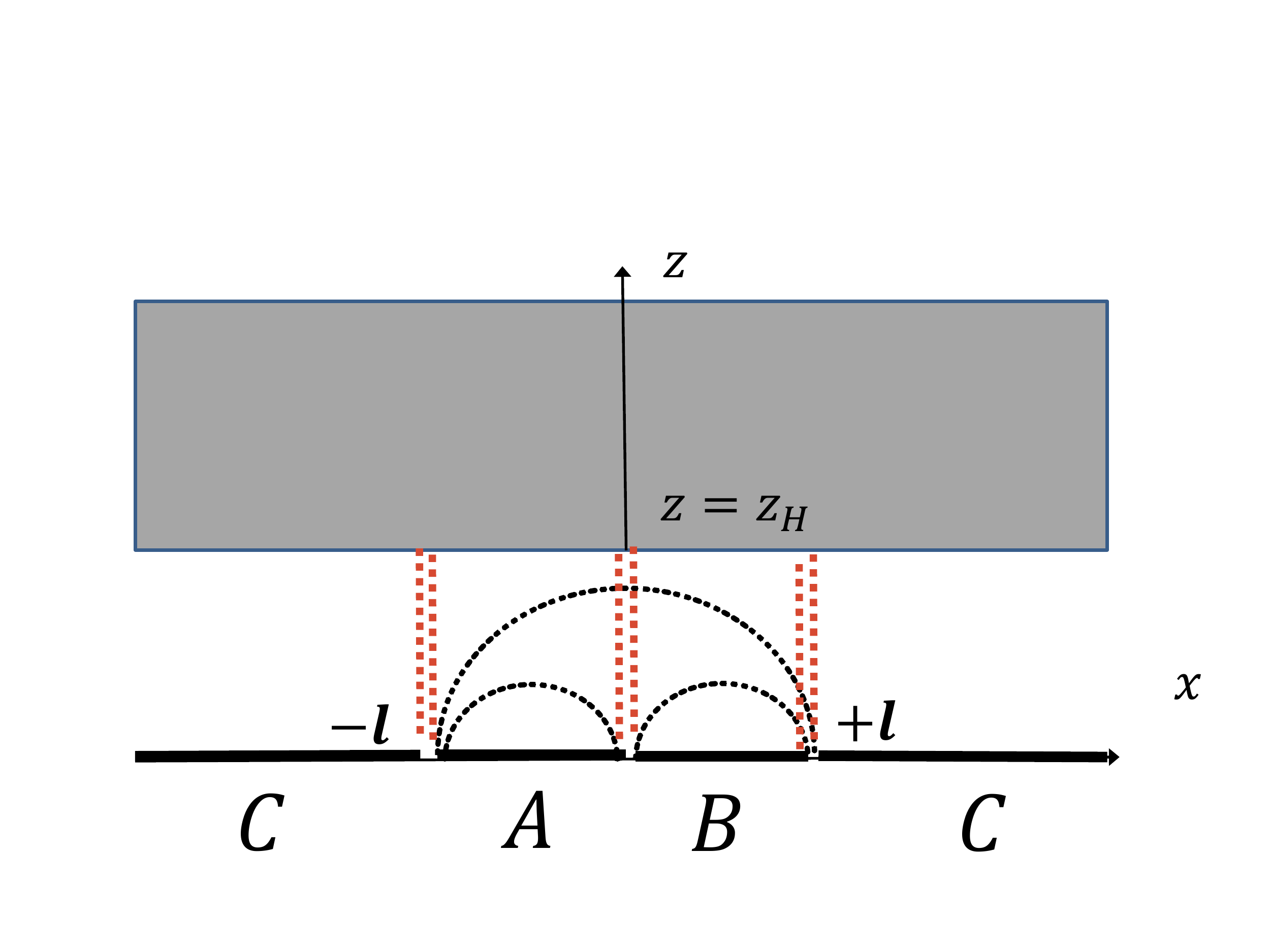}\caption{\label{fig:MEoP5} The computation of $\Delta_{W}$ in BTZ. }
\end{figure}

As studied in \citep{TU:HEoP}, the geodesic length between the boundary
and the horizon is 
\begin{equation}
L_{1}=\log\frac{\beta}{\pi\epsilon}\ ,
\end{equation}
where $\epsilon$ is the UV cutoff. The geodesic length between $(-\ell,0)$
and $(0,0)$ is \footnote{It is the same as the geodesic length between $(0,0)$ and $(\ell ,0)$.}
\begin{equation}
L_{2}=2\log\frac{\beta\sinh(\pi\ell/\beta)}{\pi\epsilon}\ ,
\end{equation}
and the geodesic length between $(-\ell,0)$ and $(\ell,0)$ is
\begin{equation}
L_{3}=2\log\frac{\beta\sinh(2\pi\ell/\beta)}{\pi\epsilon}\ .
\end{equation}
At high temperature, $\Sigma_{ABC}$ consists of six short lines of
length $L_{1}$, so the tripartite cross-section is given by 
\begin{equation}
\Delta_{W}=6L_{1}=6\log\frac{\beta}{\pi\epsilon}=:A^{(1)}\ .
\end{equation}
At low temperature, $\Sigma_{ABC}$ consists of 3 geodesic lines connecting
$(-\ell,0)$, $(0,0)$ and $(\ell,0)$, so the $\Delta_{W}$ is given
by 
\begin{equation}
\Delta_{W}=2L_{2}+L_{3}=4\log\frac{\beta\sinh(\pi\ell/\beta)}{\pi\epsilon}+2\log\frac{\beta\sinh(2\pi\ell/\beta)}{\pi\epsilon}=:A^{(2)}\ .
\end{equation}
 The tripartite entanglement wedge cross-section $\Delta_{W}$ is
thus given by 
\begin{equation}
\Delta_{W}(A:B:C)=\frac{1}{4G_{N}}\min[A^{(1)},A^{(2)}]\ .
\end{equation}
We compare $A^{(1)}$ and $A^{(2)}$ and find the critical temperature
\begin{equation}
\beta_{*}=\frac{\log\sqrt{y_{*}}}{\pi\ell}\ ,
\end{equation}
where $y_{*}$ is the positive root of 
\begin{equation}
(y+y^{-1}-2)(y-y^{-1})=8\ .
\end{equation}
Furthermore, one can confirm that there are only two phases $A^{(1)},A^{(2)}$
separated by the above critical temperature. When $\beta<\beta_{*}$,
$A^{(1)}$ is favored and when $\beta>\beta_{*}$, $A^{(2)}$ is favored.

\section{Conclusion}

\label{sec5} In this paper, we defined a generalization of entanglement
of purification to multipartite states denoted by $\Delta_{P}$, and
proved its various properties focusing on bounds described
by other entropic quantities. We demonstrate that $\Delta_{P}$ satisfies
desired properties as a generalization of $E_{P}$ and provides an upper
bound on the multipartite mutual information introduced in \citep{YHHHOS:mEsq,AHS:mEsq}.
In particular, we show that the tripartite entanglement of purification
$\Delta_{P}(A:B:C)$ is bounded from below by a sum of two different generalizations of mutual information. We also show that
for a class of tripartite quantum states which saturate the subadditivity
or the strong subadditivity, there is a closed expression of $\Delta_{P}$
in terms of entanglement entropy.

Based on the holographic conjecture of entanglement of purification
\citep{TU:HEoP,NDHZS:HEoP}, we defined a generalization of entanglement
wedge cross-section as the holographic counterpart of the multipartite entanglement
of purification $\Delta_{W}$, and show that all the properties
$\Delta_{P}$ has are indeed satisfied by $\Delta_{W}$. It leads
us to propose a new conjecture $\Delta_{P}=\Delta_{W}$ at the leading
order $O(N^{2})$ of large $N$ limit. Alternatively speaking, this
implies the naive picture of purified geometry based on tensor network
description \citep{MiyajiTakayanagi:SurfaceStateCorrespondence,MTW:OptimizationofPathIntegral1,CKMTW:OptimizationofPathIntegral2} still
works for multipartite cases. As explicit examples, we calculated
$\Delta_{W}$ for several simple setups in AdS$_{3}$/CFT$_{2}$ including
pure AdS$_{3}$ and black hole geometry.


Several future questions are in order: First, proof of $\Delta_{W}$
properties in time-dependent background geometry. The properties in
this case are expected to be very useful in understanding the dynamical
process in quantum gravity systems, since many interesting physical
set up go beyond the bi-partite pattern. The entropic inequalities
for multipartite cases were actually discussed in \citep{RW:MaximinisNotEnough}
for covariant cases, there they argued that new techniques apart from
the maximin surfaces \citep{Wall:Maximinsurfaces} are needed to
prove them for 5 or more partite cases. Second, we expect that there
is an operational interpretation for $\Delta_{P}$ just like $E_{P}$
has, such as the one based on SLOCC for multipartite qubits. Third,
looking for the new properties satisfied by $\Delta_{W}$ but not always
by $\Delta_{P}$ is certainly interesting, because these are essentially
the new constraints on holographic states, such as the conjectured
strong superadditivity of entanglement of purification in holography.
We shall report the progress in future publications.

\section*{Acknowledgements}

We thank Arpan Bhattacharyya, Masahito Hayashi, Ling-Yan Hung, Tadashi Takayanagi, Kento Watanabe, Huang-Jun Zhu
for useful conversations and comments on the draft. KU is supported
by JSPS fellowships. KU is also supported by Grant-in-Aid for JSPS Fellows No.18J22888.

\appendix
\section{The multipartite squashed entanglement and holography\label{sec4}}

There have been a lot of measures of genuine quantum entanglement
proposed. In particular, the squashed entanglement \citep{Tucci:Esq,ChristandlWinter:SquashedEntanglement}
is the most promising measure of quantum entanglement for mixed states,
as it satisfies all known desirable properties e.g. additivity. The
squashed entanglement is defined by 
\begin{equation}
E_{sq}(\rho_{AB}):=\frac{1}{2}\inf_{\rho_{ABE}}I(A:B|E),\label{eq:Def Esq}
\end{equation}
where $I(A:B|E)=I(A:BE)-I(A:E)=S_{AE}+S_{BE}-S_{ABE}-S_E$ is the conditional mutual information,
and the minimization is taken over all possible extensions $\rho_{AB}={\rm Tr}_{E}[\rho_{ABE}]$. By taking $E$ is trivial we see that by definition $E_{sq}$ is less than or equal to half of the mutual information: $E_{sq}\leq I/2$.

In \citep{YHHHOS:mEsq,AHS:mEsq}, multipartite generalizations of
$E_{sq}$ were also introduced. To define the one we are interested, we first introduce a
conditional multipartite mutual information,
\begin{equation}
I(A_{1}:\cdots:A_{n}|E)=I(A_{1}:A_{2}|E)+I(A_{1}A_{2}:A_{3}|E)+\cdots+I(A_{1}\cdots A_{n-1}:A_{n}|E).
\end{equation}
Using this, the ($q$-)multipartite
squashed entanglement for $n$-partite state $\rho_{A_{1}\cdots A_{n}}$
is defined by \footnote{Here we also define the multipartite squashed entanglemet without $\frac{1}{2}$ prefactor following our convention for $\Delta_P$.}
\begin{equation}
E_{sq}^{q}(A_{1}:\cdots:A_{n}):=\inf_{\rho_{A_{1}\cdots A_{n}E}}I(A_{1}:\cdots:A_{n}|E),
\end{equation}
where the minimization is taken over all possible extensions of $\rho_{A_{1}\cdots A_{n}}$. 

Noting that a trivial extension $\rho_{A_{1}\cdots A_{n}E}=\rho_{A_{1}\cdots A_{n}}\otimes\ket{0}\bra{0}_{E}$
gives $I(A_{1}:\cdots:A_{n}|E)=I(A_{1}:\cdots:A_{n})$, we get a generic
bound,

\begin{equation}
E_{sq}^{q}(A_{1}:\cdots:A_{n})\leq I(A_{1}:\cdots:A_{n}).\label{eq:mEsq < mMI}
\end{equation}

Combining it with the proposition \ref{prop:n-partite Delta > MI}, we get
a generic order of three different measures of multipartite correlations. 
\begin{cor}
It holds for any $n$-partite quantum states that
\begin{equation}
E_{sq}^{q}(A_{1}:\cdots:A_{n})\leq I(A_{1}:\cdots:A_{n})\leq\Delta_{P}(A_{1}:\cdots:A_{n}).
\end{equation}
\end{cor}

Note that for any pure $n$-partite state these bounds are saturated
and we have $E_{sq}^{q}=I=\Delta_{P}$. Generally, $E_{sq}^{q}\leq\Delta_{P}$ is a desirable property, since $\Delta_{P}$ is expected
to be a measure of both quantum and classical correlations while $E_{sq}^{q}$
is only of quantum ones.

\subsection{Holographic counterpart of $E_{sq}$}

The definition of squashed entanglement \eqref{eq:Def Esq} is
similar to that of the entanglement of purification \eqref{eq:Def EoP}.
Both of them use a certain type of extension, indeed, purification is a special set of extension.
This observation motivates us to seek for a holographic counterpart of $E_{sq}$ in the same spirit
of $E_{P}$ or $\Delta_{P}$. 

Let us regard a time slice of AdS as a tensor network which describes a quantum state of CFT. In a gravity background with a tensor network description, one can define a pure or mixed state for any codimension two convex surface, called the surface/state correspondence \citep{MiyajiTakayanagi:SurfaceStateCorrespondence,MTW:OptimizationofPathIntegral1,CKMTW:OptimizationofPathIntegral2}. Let us now take use of this picture to study what the holographic counterpart of squashed entanglement could be if it indeed exists\footnote{We are grateful to Tadashi Takayanagi for lots of fruitful comments on the following discussion.}.

First, we consider a class of extensions from $\rho_{AB}$ to $\rho_{ABE}$ which have classical gravity duals described by tensor networks. The extensions are not necessarily on the original AdS boundary, but each extended geometry should include the entanglement wedge $M_{AB}$ and its boundary should be convex \citep{MiyajiTakayanagi:SurfaceStateCorrespondence}. We assume that there exists an optimal extension for (\ref{eq:Def Esq}) in this class. Then we find that all nontrivial such extensions give $I(A:B|E)\geq I(A:B)$. Indeed this is equivalent to say that the holographic tripartite information is always negative: $\tilde{I}_3(A:B:E)\leq0$. This eventually leads that the optimal solution in (\ref{eq:Def Esq}) is given by half of the mutual information in holography
\begin{equation}
E_{sq}(A:B)=\frac{1}{2}I(A:B)\ . \label{eq: Esq=I/2}
\end{equation}
In the remaining text, we illustrate our argument in details.

To compute the squashed entanglement (\ref{eq:Def Esq}), we want to consider how much correlation one can reduce between $A$ and $B$ by knowing the ancillary system $E$. This is appropriately quantified by the tripartite information
\begin{equation}
\tilde{I}_3(A:B:E)=I(A:B)-I(A:B|E).
\end{equation}
This quantity has been widely studied in holography. In particular, the so called monogamy of mutual information (MMI) shows that $\tilde{I}_3(A:B:E)$ is always negative in holography in the case $E$ is on the original AdS boundary \citep{HHM:MonogamyofHMI}
\begin{equation}
\label{theMMI}
I(A:BE)\geq I(A:B)+I(A:E)\ \Leftrightarrow \ \tilde{I}_3(A:B:E)\leq 0.
\end{equation}
Hence one can never reduce the correlation between $A$ and $B$ by knowing the extension $E$ located on the original AdS boundary, i.e. $I(A:B|E)\geq I(A:B)$ holds in this case (Fig.\ref{fig:MEoP11}). Note that (\ref{theMMI}) is a characteristic property of holographic states, since it is not generally satisfied by arbitrary tripartite quantum states.

\begin{figure}[H]
\centering{}\includegraphics[scale=0.33]{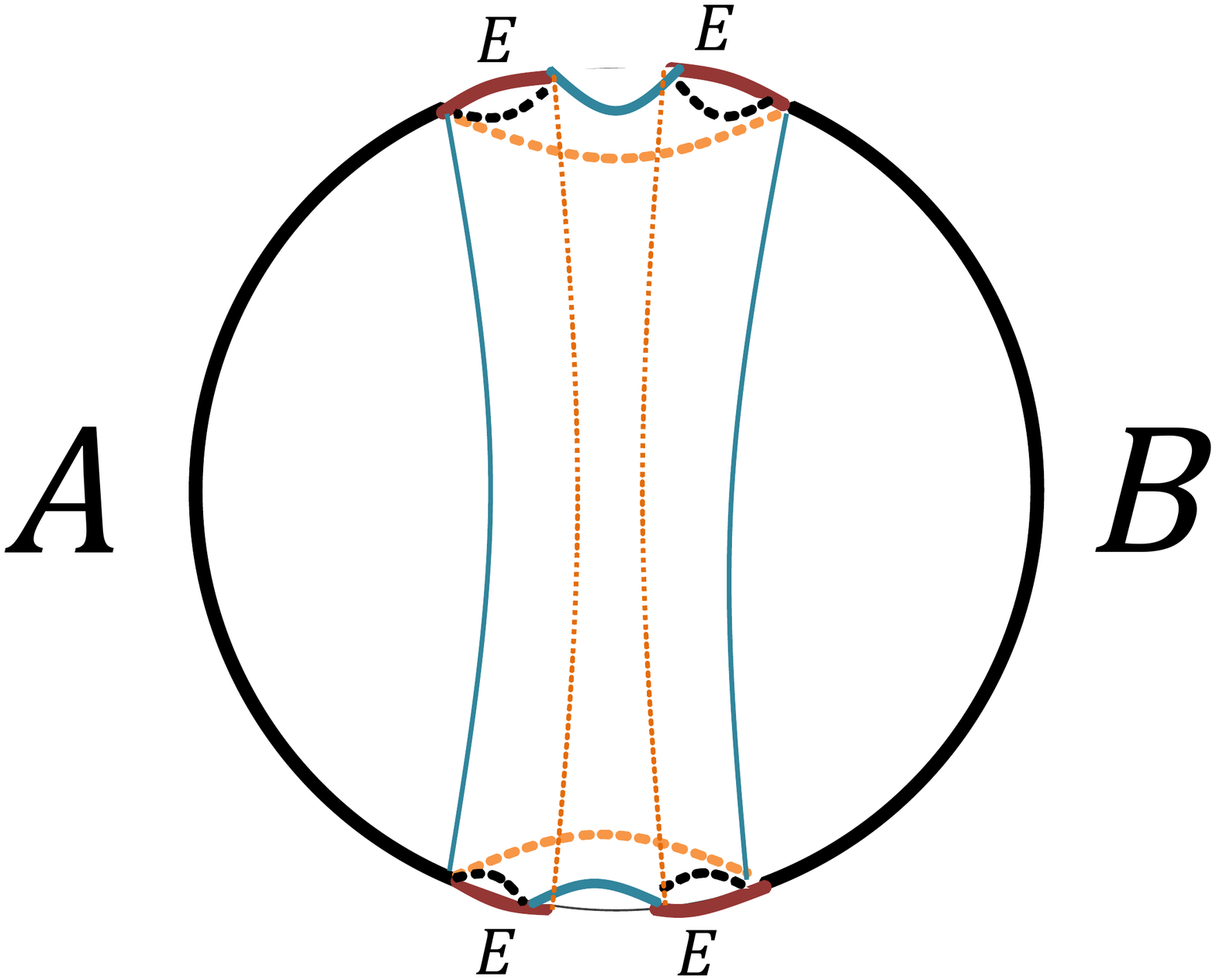}\caption{\label{fig:MEoP11} A proof of $I(A:B|E)\geq I(A:B)$ for an extension $\rho_{ABE}$, where $E$ is on the original boundary \citep{HHM:MonogamyofHMI}. The difference $I(A:B|E)-I(A:B)$ is given by the area of dashed orange codimension-2 surfaces minus that of solid blue ones, which is clearly non-negative by the minimal property of RT-surfaces. }
\end{figure}

Now, we would like to assert that this type of monogamy (\ref{theMMI}) is still true in more general holographic states we consider, as long as the entanglement entropy is still given by a minimal area functional in dual geometry\footnote{One can replace in this argument the area functional to any geometrical functional as long as it is extensive \citep{HHM:MonogamyofHMI}.  In particular, the authors of \citep{HHM:MonogamyofHMI} showed that higher curvature corrections does not affect the MMI, and conjectured that it will be a common property of all large-N field theories.}. This is because the proof of MMI in \citep{HHM:MonogamyofHMI} does not rely on any peculiarity of asymptotic AdS boundary. In particular, it does not concern whether the boundary is at the asymptotic infinity or not. This fact temps us to assume that all extensions with classical gravity duals satisfy the monogamy of mutual information (or equivalentally $\tilde{I}_3(A:B:E)\leq 0$).
This situation can be graphically illustrated for example in Fig.\ref{fig:MEoP12}.

\begin{figure}[H]
\centering{}\includegraphics[scale=0.33]{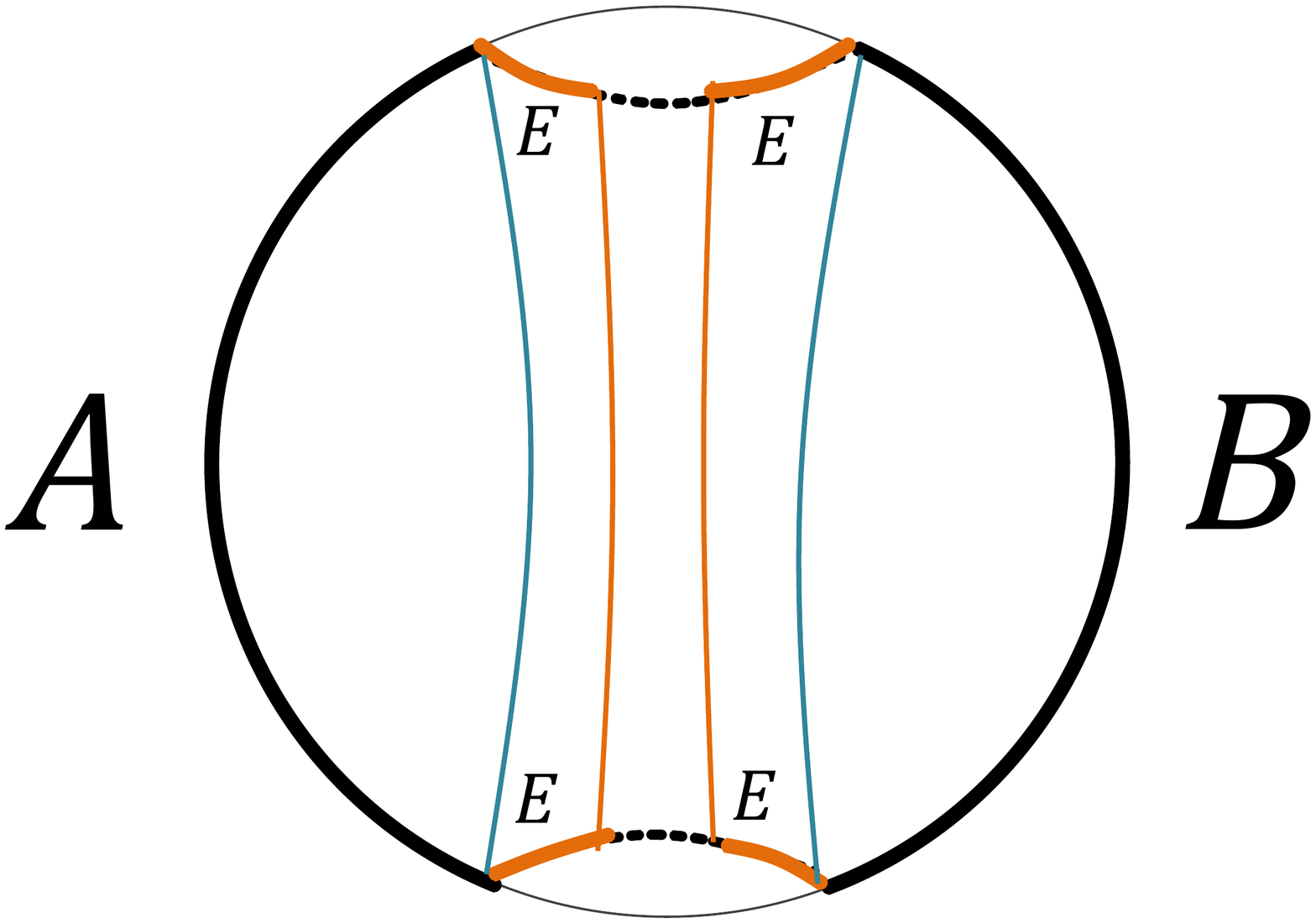}\caption{\label{fig:MEoP12} A proof of $I(A:B|E)\geq I(A:B)$ for an extension $\rho_{ABE}$, where $E$ is on the minimal purification (i.e. RT-surface of $S_{AB}$). The difference $I(A:B|E)-I(A:B)$ in this case is given by the area of the orange codimension-2 surfaces minus that of the blue ones, which is certainly non-negative. For simplicity here we assumed that $E$ is relatively small, but the same result is true for large $E$ (i.e. near the purification).}
\end{figure}

From the above argument we conjecture that one cannot reduce the correlation between $A$ and $B$ by measuring an ancillary system $E$ i.e. $I(A:B|E)\geq I(A:B)$ holds for any extension with classical gravity dual. With the assumption that in (\ref{eq:Def Esq}) there exists an optimal extension with a classical geometrical description, the squashed entanglement in holography will be given by just half of the mutual information (at $O(N^2)$):
\begin{equation}
E_{sq}(A:B)=\frac{1}{2}I(A:B)\ ,
\end{equation}
which can clearly be achieved by a trivial extension. 
This tells that the inequality $E_{sq}(A:B)\leq I(A:B)/2$ is saturated in holography\footnote{This saturation can also be observed in a tensor network model of holography \citep{NezamiWalter:MultiEntinStabilizerTN}.}.

As a consistency check of this saturation conjecture, one can test whether the holographic properties of them (such as the additivity) coincide or not. 
As mentioned above, the mutual information does not always satisfy the monogamy, while the squashed entanglement does for any tripartite state \citep{KoashiWinter:monogamy}. The fact that the mutual information becomes monogamous in holography actually provides a non-trivial check of (\ref{eq: Esq=I/2}). One might also regard (\ref{eq: Esq=I/2}) as the origin of the MMI in holography.\\

Note that the saturation (\ref{eq: Esq=I/2}) is already implicated in \citep{HHM:MonogamyofHMI}, motivated by the MMI itself. \\

Moreover, we expect that this kind of saturation also happens for
multipartite cases. 
Namely, we are tempted to test the saturation of \eqref{eq:mEsq < mMI} in
holography (at $O(N^2)$): 
\begin{equation}
E_{sq}^{q}(A_{1}:\cdots:A_{n})=^{?}I(A_{1}:\cdots:A_{n}).\label{eq:mEsq=00003D00003DmMI}
\end{equation}
One of such non-trivial tests is given by the strong superadditivity of multipartite
squashed entanglement \citep{YHHHOS:mEsq}:
\begin{equation}
E_{sq}^{q}(A_{1}B_{1}:\cdots:A_{n}B_{n})\geq E_{sq}^{q}(A_{1}:\cdots:A_{n})+E_{sq}^{q}(B_{1}:\cdots:B_{n}).
\end{equation}
This is true in any $2n$-partite state $\rho_{A_{1}B_{1}\cdots A_{n}B_{n}}$.
Note that the multipartite mutual information does not satisfy this property in
general. For example, if one considers a quantum state 
\begin{equation}
\rho_{A_{1}B_{1}A_{2}B_{2}A_{3}B_{3}}=\frac{1}{\sqrt{2}}(\ket{000000}\bra{000000}_{A_{1}B_{1}A_{2}B_{2}A_{3}B_{3}}+\ket{111111}\bra{1111111}_{A_{1}B_{1}A_{2}B_{2}A_{3}B_{3}}),
\end{equation}
This will lead to $I(A_{1}B_{1}:A_{2}B_{2}:A_{3}B_{3})=I(A_{1}:A_{2}:A_{3})=I(B_{1}:B_{2}:B_{3})$, which clearly violates the inequality.

However, in holographic CFTs, one can show that the multipartite mutual information does
satisfy the strong superadditivity. This comes as follows: 
\begin{align}
I(A_{1}B_{1}:\cdots:A_{n}B_{n}) & =I(A_{1}B_{1}:A_{2}B_{2})+I(A_{1}B_{1}A_{2}B_{2}:A_{3}B_{3})+\cdots \nonumber \\ &\ \ \ \  \  \cdots+I(A_{1}B_{1}A_{2}B_{2}\cdots A_{n-1}B_{n-1}:A_{n}B_{n})\nonumber \\
 & \geq I(A_{1}:A_{2})+I(A_{1}A_{2}:A_{3})+\cdots+I(A_{1}A_{2}\cdots A_{n-1}:A_{n})\nonumber \\
 & \ \ \ +I(B_{1}:B_{2})+I(B_{1}B_{2}:B_{3})+\cdots+I(B_{1}B_{2}\cdots B_{n-1}:B_{n})\nonumber \\
 & =I(A_{1}:\cdots:A_{n})+I(B_{1}:\cdots:B_{n}),
\end{align}
where we have used the monogamy of holographic mutual information recursively.
This observation gives us further evidence for our conjectured saturation in holography
\eqref{eq:mEsq=00003D00003DmMI}.

To our best knowledge, there is no counter example for such saturation conjecture. 
If there is a property which $E_{sq}$ always satisfies but $I$ does not always, then one can try to test whether it holds for holographic mutual information. If it does, this can be considered as additional positive evidence for our conjecture.\footnote{Note that related properties about holographic entanglement entropy
are intensively studied in \citep{BNOSS:EntropyCone} for multipartite
setups.}



\begin{thebibliography}{}

\bibitem{Vidal:2002rm} 
  G.~Vidal, J.~I.~Latorre, E.~Rico and A.~Kitaev,
  ``Entanglement in quantum critical phenomena,''
  Phys.\ Rev.\ Lett.\  {\bf 90}, 227902 (2003)
  quant-ph/0211074.
  
\bibitem{Kitaev:2005dm} 
  A.~Kitaev and J.~Preskill,
  ``Topological entanglement entropy,''
  Phys.\ Rev.\ Lett.\  {\bf 96}, 110404 (2006)
  hep-th/0510092.
  
\bibitem{Levin:2006} 
  M.~Levin and X.G.~Wen,
  ``Detecting topological order in a ground state wave function,''
  Phys.\ Rev.\ Lett.\  {\bf 96}, 110405 (2006)
  cond-mat/0510613.


\bibitem{Casini:2004bw} 
  H.~Casini and M.~Huerta,
  ``A Finite entanglement entropy and the c-theorem,''
  Phys.\ Lett.\ B {\bf 600}, 142 (2004)
  hep-th/0405111.

\bibitem{Calabrese:2004eu}
  P.~Calabrese, J.~Cardy,
  ``Entanglement Entropy and Quantum Field Theory,''
  J.~Stat.~Mech. {\bf 0406} 06002 (2004)
  hep-th/0405152.
 
\bibitem{Casini:2012ei} 
  H.~Casini and M.~Huerta,
  ``On the RG running of the entanglement entropy of a circle,''
  Phys.\ Rev.\ D {\bf 85}, 125016 (2012)
  arXiv:1202.5650 [hep-th].
  
\bibitem{Casini:2008cr} 
  H.~Casini,
  ``Relative entropy and the Bekenstein bound,''
  Class.\ Quant.\ Grav.\  {\bf 25}, 205021 (2008)
  arXiv:0804.2182 [hep-th].

  
\bibitem{Hung:2018rhg} 
  L.~Y.~Hung, Y.~S.~Wu and Y.~Zhou,
  ``Linking Entanglement and Discrete Anomaly,''
  arXiv:1801.04538 [hep-th].


\bibitem{RT:RT-formula}
 S.~Ryu and T.~Takayanagi,
 ``Holographic Derivation of Entanglement Entropy from AdS/CFT,"
 Phys.\ Rev.\ Lett.\  {\bf 96} (2006) 181602  [hep-th/0603001];
  ``Aspects of Holographic Entanglement Entropy," JHEP {\bf 0608} (2006) 045
  hep-th/0605073.

\bibitem{HRT:HRT-formula}
  V.~E.~Hubeny, M.~Rangamani and T.~Takayanagi,
  ``A Covariant holographic entanglement entropy proposal,'' 
  JHEP {\bf 0707} (2007) 062 
  arXiv:0705.0016 [hep-th]. 

\bibitem{Maldacena:AdS/CFT}
  J.~M.~Maldacena,
  ``The Large N Limit of Superconformal Field Theories and Supergravity,''
  Adv.\ Theor.\ Math.\ Phys.\  {\bf 2} (1998) 231
  [Int.\ J.\ Theor.\ Phys.\  {\bf 38} (1999) 1113]
  hep-th/9711200.

\bibitem{Swingle:TensorNetwork}
  B.~Swingle,
  ``Entanglement Renormalization and Holography,'' 
  Phys. Rev. {\bf D 86}, 065007 (2012),
  arXiv:0905.1317 [cond-mat.str-el].

\bibitem{VanRaamsdonk:Buildingup1}
  M.~Van Raamsdonk,
  ``Comments on quantum gravity and entanglement,''
  \href{https://arxiv.org/abs/0907.2939}{arXiv:0907.2939} [hep-th];
  ``Building up spacetime with quantum entanglement,''
  Gen.\ Rel.\ Grav.\  {\bf 42} (2010) 2323
   [Int.\ J.\ Mod.\ Phys.\ D {\bf 19} (2010) 2429],
  \href{https://arxiv.org/abs/1005.3035}{arXiv:1005.3035} [hep-th].

\bibitem{MaldacenaSusskind:ERequalsEPR}
  J.~Maldacena, L.~Susskind, 
  ``Cool horizons for entangled black holes," 
  Fortsch. Phys. 61 (2013) 781-811,
  arXiv:1306.0533 [hep-th].

\bibitem{FGHMR:LinearEinsteinequationfrom1stLaw}
  T.~Faulkner, M.~Guica, T.~Hartman, R.~C.~Myers, M.~V.~Raamsdonk,
  ``Gravitation from Entanglement in Holographic CFTs,"
  JHEP {\bf 03} (2014) 051,
  arXiv:1312.7856 [hep-th].

\bibitem{MiyajiTakayanagi:SurfaceStateCorrespondence}
  M.~Miyaji and T.~Takayanagi,
  ``Surface/State Correspondence as a Generalized Holography,''
  PTEP {\bf 2015} (2015) no.7,  073B03
  arXiv:1503.03542 [hep-th].

\bibitem{MTW:OptimizationofPathIntegral1}
  M.~Miyaji, T.~Takayanagi and K.~Watanabe,
  ``From Path Integrals to Tensor Networks for AdS/CFT,''
  Phys.\ Rev.\ D {\bf 95} (2017) no.6,  066004,
  arXiv:1609.04645 [hep-th].

\bibitem{CKMTW:OptimizationofPathIntegral2}
  P.~Caputa, N.~Kundu, M.~Miyaji, T.~Takayanagi and K.~Watanabe,
  ``AdS from Optimization of Path-Integrals in CFTs,''
  Phys.\ Rev.\ Lett. {\bf 119} (2017) 071602,
  arXiv:1703.00456 [hep-th],
  ``Liouville Action as Path-Integral Complexity: From Continuous Tensor Networks to AdS/CFT,''
  JHEP {\bf 11} (2017) 097 
  arXiv:1706.07056 [hep-th].

\bibitem{PYHP:HaPPYcode}
  F.~Pastawski, B.~Yoshida, D.~Harlow and J.~Preskill,
 ``Holographic quantum error-correcting codes: Toy models for the bulk/boundary correspondence,''
  JHEP {\bf 1506} (2015) 149,
  arXiv:1503.06237 [hep-th].

\bibitem{FreedmanHeadrick:BitThreads}
  M.~Freedman and M.~Headrick,
  ``Bit threads and holographic entanglement,''
  Commun.\ Math.\ Phys.\  {\bf 352} (2017) no.1,  407
  arXiv:1604.00354 [hep-th].


\bibitem{THLD:EoP}
B.~M.~Terhal, M.~Horodecki, D.~W.~Leung and D.~P.~DiVincenzo,
``The entanglement of purification,''
J.\ Math.\ Phys.\ {\bf 43} (2002) 4286,
quant-ph/0202044.

\bibitem{TU:HEoP}
T.~Takayanagi and K.~Umemoto,
``Holographic Entanglement of Purification,''
  Nat. Phys. {\bf 14}, 573-577 (2018)
  arXiv:1708.09393 [hep-th].

\bibitem{NDHZS:HEoP}
  P.~Nguyen, T.~Devakul, M.~G.~Halbasch, M.~P.~Zaletel and B.~Swingle,
  ``Entanglement of purification: from spin chains to holography,''
  JHEP {\bf 1801} (2018) 098
  arXiv:1709.07424 [hep-th].

\bibitem{BaoHalpern:HEoPgeneralization}
  N.~Bao and I.~F.~Halpern,
  ``Holographic Inequalities and Entanglement of Purification,''
  JHEP {\bf 03} (2018) 006,
  arXiv:1710.07643 [hep-th].

\bibitem{ATU:EoPfreeQFTs}
  A.~Bhattacharyya, T.~Takayanagi, K.~Umemoto,
  ``Entanglement of Purification in Free Scalar Field Theories,"
   JHEP {\bf 1804} (2018) 132,
  arXiv:1802.09545 [hep-th].

\bibitem{HiraiTamaokaYokoya:TowardsEoPinCFTs}
  H.~Hirai, K.~Tamaoka, T.~Yokoya,
  ``Towards Entanglement of Purification for Conformal Field Theories,"
  arXiv:1803.10539 [hep-th].

\bibitem{EGP:WEReconstructionEoP}
  R.~Espíndola, A.~Guijosa, J.~F.~Pedraza,
  ``Entanglement Wedge Reconstruction and Entanglement of Purification,"
  arXiv:1804.05855 [hep-th].

\bibitem{NRS:PullingBoundaryintoBulk}
  Y.~Nomura, P.~Rath, N.~Salzetta, 
  ``Pulling the Boundary into the Bulk,"
  arXiv:1805.00523 [hep-th].

\bibitem{HHHH:QuantumEntanglement}
  R.~Horodecki, P.~Horodecki, M.~Horodecki and K.~Horodecki,
  ``Quantum entanglement,''
  Rev.\ Mod.\ Phys.\  {\bf 81} (2009) 865
  quant-ph/0702225.

\bibitem{HHM:MonogamyofHMI}
  P.~Hayden, M.~Headrick and A.~Maloney,
  ``Holographic Mutual Information is Monogamous,''
  Phys.\ Rev.\ D {\bf 87} (2013) no.4,  046003
  arXiv:1107.2940 [hep-th].

\bibitem{BHMMR:MultiboundaryWormholes}
  V.~Balasubramanian, P.~Hayden, A.~Maloney, D.~Marolf, S.~F.~Ross,
  ``Multiboundary Wormholes and Holographic Entanglement,"
  Classical and Quantum Gravity, {\bf 31} (18):185015, 2014,
  arXiv:1406.2663 [hep-th].

\bibitem{NezamiWalter:MultiEntinStabilizerTN}
  S.~Nezami, M.~Walter,
  ``Multipartite Entanglement in Stabilizer Tensor Networks,"
  arXiv:1608.02595 [quant-ph].


\bibitem{HQRY:ChaosinQCs}
  P.~Hosur, X.~L.~Qi, D.~A.~Roberts, B.~Yoshida,
  ``Chaos in quantum channels,"
  JHEP {\bf 02} (2016) 004,
  arXiv:1511.04021 [hep-th].


\bibitem{Rota:TripartiteMI}
  M.~Rota,
  ``Tripartite information of highly entangled states,"
  JHEP {\bf 04} (2016) 075 
  [arXiv:1512.03751 [hep-th].

\bibitem{MTV:NpartiteHMI}
  S.~Mirabi, M.~Reza Tanhayi and R.~Vazirian,
  ``On the Monogamy of Holographic n-partite Information,"
  Phys. Rev. D {\bf 93}, 104049 (2016),
  arXiv:1603.00184 [hep-th].

\bibitem{AMT:NpartiteHMI}
  M.~Alishahiha, M.~R.~M.~Mozaffar, M.~Reza~Tanhayi,
  ``On the Time Evolution of Holographic n-partite Information,"
  JHEP {\bf 09} (2015) 165
  arXiv:1406.7677 [hep-th].

\bibitem{YHHHOS:mEsq}
  D.~Yang, K.~Horodecki, M.~Horodecki, P.~Horodecki, J.~Oppenheim, W.~Song,
  ``Squashed entanglement for multipartite states and entanglement measures based on the mixed convex roof,"
  IEEE Trans. Inf. Theory {\bf 55}, 3375 (2009),
  arXiv:0704.2236 [quant-ph].

\bibitem{AHS:mEsq}
  D.~Avis, P.~Hayden, I.~Savov,
  ``Distributed Compression and Multiparty Squashed Entanglement"
  J. Phys. A {\bf 41} (2008) 115301,
  arXiv:0707.2792 [quant-ph].


\bibitem{CKNR:EntanglementWedge}
B.~Czech, J.~L.~Karczmarek, F.~Nogueira and M.~Van Raamsdonk,
``The Gravity Dual of a Density Matrix,''  Class.\ Quant.\ Grav.\  {\bf 29} (2012) 155009 
   arXiv:1204.1330 [hep-th].

\bibitem{Wall:Maximinsurfaces}
  A.~C.~Wall,
  ``Maximin Surfaces, and the Strong Subadditivity of the Covariant Holographic Entanglement Entropy,''
  Class.\ Quant.\ Grav.\  {\bf 31} (2014) no.22,  225007 
  arXiv:1211.3494 [hep-th].

\bibitem{HHLR:EntanglementWedge}
  M.~Headrick, V.~E.~Hubeny, A.~Lawrence and M.~Rangamani,
  ``Causality and holographic entanglement entropy,''
  JHEP {\bf 1412} (2014) 162
  arXiv:1408.6300 [hep-th].

\bibitem{Tucci:Esq}
  R.~R.~Tucci,
  ``Entanglement of Distillation and Conditional Mutual Information,"
  quant-ph/0202144.

\bibitem{ChristandlWinter:SquashedEntanglement}
  M.~Christandl and A.~Winter,
  ````Squashed entanglement'': An additive entanglement measure,''
  J.\ Math.\ Phys.\ {\bf 45} (2004) 829,
  quant-ph/0308088.

\bibitem{KoashiWinter:monogamy}
  M.~Koashi and A.~Winter,
  ``Monogamy of entanglement and other correlations,''
  Phys.\ Rev.\ A {\bf 69} (2004) 022309,
  quant-ph/0310037.

\bibitem{BaoHalpern:multipartiteEoP}
  N.~Bao, I.~F.~Halpern,
  ``Conditional and Multipartite Entanglements of Purification and Holography,"
  arXiv:1805.00476 [hep-th].

\bibitem{BagchiPati:EoPproperties}
  S. ~Bagchi and A. ~K. ~Pati,
  ``Monogamy, polygamy, and other properties of entanglement of purification'',
  Phys.\ Rev.\  A{\bf 91} (2015) 042323
  arXiv:1502.01272.


\bibitem{BNOSS:EntropyCone}
  N.~Bao, S.~Nezami, H.~Ooguri, B.~Stoica, J.~Sully, M.~Walter,
  ``The Holographic Entropy Cone,"
  JHEP {\bf 09} (2015) 130,
  arXiv:1505.07839 [hep-th].

\bibitem{RW:MaximinisNotEnough}
  M.~Rota, S.~J.~Weinberg,
  ``Maximin is Not Enough,"
  Phys. Rev. D {\bf 97}, 086013 (2018),
  arXiv:1712.10004 [hep-th].




 \end{thebibliography}
\end{document}